\documentclass{article}

    \PassOptionsToPackage{numbers, compress}{natbib}

 \usepackage[preprint]{neurips_2026}

\usepackage[utf8]{inputenc} %
\usepackage[T1]{fontenc}    %
\usepackage{hyperref}       %
\usepackage{url}            %
\usepackage{booktabs}       %
\usepackage{amsfonts}       %
\usepackage{nicefrac}       %
\usepackage{microtype}      %
\usepackage{xcolor}         %

\usepackage[textsize=tiny]{todonotes}

\usepackage{amsmath,amsfonts,amsthm,bm,bbm,xfrac,dsfont,stmaryrd,xfrac}

\DeclareMathOperator*{\bigtimes}{\vartimes}

\def\eqref#1{equation~\ref{#1}}

\def\peqref#1{(\ref{#1})}

\def\appxref#1{Appx.~\ref{#1}}

\def\1{\bm{1}}

\def\va{{\bm{a}}}

\def\vs{{\bm{s}}}

\def\vx{{\bm{x}}}

\def\vz{{\bm{z}}}

\DeclareMathAlphabet{\mathsfit}{\encodingdefault}{\sfdefault}{m}{sl}
\SetMathAlphabet{\mathsfit}{bold}{\encodingdefault}{\sfdefault}{bx}{n}

\DeclareMathOperator*{\argmax}{arg\,max}
\DeclareMathOperator*{\argmin}{arg\,min}

\makeatletter
\newtheorem*{rep@theorem}{\rep@title}
\newcommand{\newreptheorem}[2]{%
\newenvironment{rep#1}[1]{%
 \def\rep@title{#2 \ref{##1}}%
 \begin{rep@theorem}}%
 {\end{rep@theorem}}}
\makeatother

\newtheorem{theorem}{Theorem}
\newreptheorem{theorem}{Theorem}
\newtheorem*{theorem*}{Theorem}
\newtheorem{proposition}{Proposition}
\newreptheorem{proposition}{Proposition}
\newtheorem{corollary}{Corollary}
\newreptheorem{corollary}{Corollary}

\newreptheorem{lemma}{Lemma}

\newreptheorem{remark}{Remark}

\newtheorem{definition}{Definition}

{\end{list}
}
\usepackage{csquotes}

\usepackage{subcaption}

\usepackage[capitalize,noabbrev]{cleveref}

\newcommand{\tikzcircle}[2][black,fill=black]{\tikz[baseline=-0.5ex]\draw[#1,radius=#2] (0,0) circle ;}%

\usepackage{algorithm,algorithmic}
\makeatletter
\newcommand{\alglabel}[1]{%
  \addtocounter{ALC@line}{-1}%
  \refstepcounter{ALC@line}%
  \label{#1}%
}
\makeatother

\usepackage{subcaption, graphicx}

\usepackage{multirow}

\usepackage{comment}

\usepackage{enumitem}

\usepackage{calc}
\usepackage{ifthen}
\usepackage{tikz}
\usetikzlibrary{calc,arrows.meta} %
\usetikzlibrary{positioning} %
\usetikzlibrary{decorations.pathreplacing,patterns}
\usepackage[normalem]{ulem}

\newcommand{\tikzmark}[1]{\tikz[overlay,remember picture] \node (#1) {};}

\usepackage{twemojis}
\newcommand{\emoji}[1]{\raisebox{-0.1em}{\scalebox{1.5}{\twemoji{#1}}}}
\usepackage{mdframed} %

\definecolor{gdmblue}{HTML}{4285f4} %
\definecolor{gdmred}{HTML}{ea4335} %
\definecolor{maroon}{HTML}{FFA500}%
\definecolor{navyblue}{HTML}{008080}%

\colorlet{goldcolor}{orange!70!yellow}
\colorlet{silvercolor}{gray!60}
\colorlet{bronzecolor}{brown!60!black}

\newcommand{\gold}{\tikz[scale=0.15]\fill[goldcolor, draw=black] circle(0.5);}
\newcommand{\silver}{\tikz[scale=0.15]\fill[silvercolor, draw=black] circle(0.5);}
\newcommand{\bronze}{\tikz[scale=0.15]\fill[bronzecolor, draw=black] circle(0.5);}

\newcommand{\br}{\texttt{BR}}

\newcommand{\tpzs}{two-player, zero-sum}

\newcommand{\pref}{\rho}
\newcommand{\wtl}{\texttt{wtl}}
\newcommand{\target}{\mu}

\newcommand{\erock}{\emoji{rock}}
\newcommand{\epaper}{\emoji{page facing up}}
\newcommand{\escissors}{\emoji{scissors}}

\newcommand{\bear}{Bear} %
\newcommand{\ebear}{\emoji{bear}} %
\newcommand{\bearlabel}{\ebear~\bear} %

\newcommand{\dog}{Dog}
\newcommand{\edog}{\emoji{dog face}}
\newcommand{\doglabel}{\edog~\dog} %

\newcommand{\camel}{Rabbit}
\newcommand{\ecamel}{\emoji{rabbit face}}
\newcommand{\camellabel}{\ecamel~\camel}

\newcommand{\goose}{Frog}
\newcommand{\egoose}{\emoji{frog}}
\newcommand{\gooselabel}{\egoose~\goose}

\newcommand{\elephant}{Koala}
\newcommand{\eelephant}{\emoji{koala}}

\newcommand{\turtle}{Chicken}
\newcommand{\eturtle}{\emoji{chicken}}

\newcommand{\zebra}{Lion}
\newcommand{\ezebra}{\emoji{lion}}

\newcommand{\dolphin}{Pig}
\newcommand{\edolphin}{\emoji{pig face}}

{}
{}
{}
{}
{}

\title{Nash without Numbers: A Social Choice Approach to Mixed Equilibria in Context-Ordinal Games}

\author{%
  Ian Gemp \\
  Google DeepMind \\
  New York, NY, USA \\
  \texttt{imgemp@google.com} \\
  \And
  Crystal Qian \\
  Google DeepMind \\
  New York, NY USA \\
  \texttt{cjqian@google.com} \\
  \AND
  Marc Lanctot \\
  Google DeepMind \\
  Montreal, CA \\
  \texttt{lanctot@google.com} \\
  \And
  Kate Larson \\
  Google DeepMind \& \\
  University of Waterloo \\
  Waterloo, CAN \\
  \texttt{katelarson@google.com}
}

\begin{document}

\maketitle

\begin{abstract}
Nash equilibrium serves as a fundamental mathematical tool in economics and game theory. However, it classically assumes knowledge of player utilities, whereas economics generally regards preferences as more fundamental. To leverage equilibrium analysis in strategic scenarios, one must first elicit numerical utilities consistent with player preferences, a delicate and time-consuming process. In this work, we forgo precise utilities and generalize the Nash equilibrium to a setting where we only assume a player is capable of providing an \emph{ordinal} ranking of their actions within the \emph{context} of other players' joint actions. The key technical challenge is to rethink the definition of a \emph{best-response}. While the classical definition identifies actions maximizing expected payoff, we naturally look towards social choice theory for how to aggregate preferences to identify the most preferred actions. We define this generalized notion of a \emph{context-ordinal} Nash equilibrium, establish its existence under mild conditions on aggregation methods, introduce notions of regularization, approximation, and regret, explore complexity for simple settings, and develop learning rules for computing such equilibria. In doing so, we provide a generalization of Nash equilibrium and demonstrate its direct applicability to elicited preferences in human experiments.
\end{abstract}

\section{Introduction}\label{sec:intro}

Game theory seeks to define rational (utility-maximizing) behavior in the presence of rational co-players. However, not all strategic scenarios admit precise numerical utilities.
In elections, for instance, voters may strategically cast ballots in order to achieve desired election results. One can imagine pursuing tactical voting without ever ascribing any precise numerical value to each electoral outcome. We later study human data from such settings in experiments.

The dominant solution concept in game theory is the notion of a Nash equilibrium (NE), a strategy profile from which no single player has any incentive to deviate~\citep{nash1950equilibrium}. Traditionally, an incentive to deviate would mean that a player has an opportunity to take an action that achieves higher expected utility.
However, how can one compute an \emph{expected} utility in a setting where numerical utilities are not available?

\begin{figure}
    \centering
    \begin{tikzpicture}[scale=1.5]
    \draw[thick, fill=black!10] (0,0) -- (0:1) arc (0:90:1) -- cycle;
    \node at (45:0.62) {\textcolor{gdmred}{25\%} \erock};
    \draw[thick, fill=black!10] (0,0) -- (90:1) arc (90:198:1) -- cycle;
    \node at (144:0.55) {\textcolor{gdmred}{30\%} \epaper};
    \draw[thick, fill=black!10] (0,0) -- (198:1) arc (198:360:1) -- cycle;
    \node at (279:0.3) {\textcolor{gdmred}{45\%} \escissors};
    \node[anchor=west] (ballot1) at (1.15, 0.45) {%
        \textcolor{gdmblue}{$\{$\epaper$\succ$\erock$\succ$\escissors$\}$}%
    };
    \draw[thin, black!30] (25:1.02) -- (ballot1.west);
    \node[anchor=east] (ballot2) at (-1.2, 0.3) {%
        \textcolor{gdmblue}{$\{$\escissors$\succ$\epaper$\succ$\erock$\}$}%
    };
    \draw[thin, black!30] (135:1.02) -- (ballot2.east);
    \node[anchor=west] (ballot3) at (1.15, -0.5) {%
        \textcolor{gdmblue}{$\{$\erock$\succ$\escissors$\succ$\epaper$\}$}%
    };
    \draw[thin, black!30] (330:1.02) -- (ballot3.west);
    \node[text=gdmred] at (0,-0.6) {Opp. Strategy};
    \node[text=black, text width=2cm] at (3.5,0.45) {\textcolor{gdmblue}{Our vote} when \textcolor{gdmred}{opp. plays} \erock};
    \end{tikzpicture}
    \caption{An NE is a strategy profile where each player best responds to its co-players. In a game without payoffs, but where players can rank their actions, we use social choice (voting) theory to define a best response. Consider playing an \textcolor{gdmred}{opponent} in rock-paper-scissors (\erock,\epaper,\escissors); their mixed strategy is \textcolor{gdmred}{$[25\%, 30\%, 45\%]$}. When the opponent plays, e.g., \erock, \textcolor{gdmblue}{our} rank vote over \textcolor{gdmblue}{our} own actions is \textcolor{gdmblue}{$\{$\epaper$\succ$\erock$\succ$\escissors$\}$}. Imagine a population of votes with representation proportional to the opponent's mixed strategy\textemdash \textcolor{gdmred}{$25\%$} of the votes are \textcolor{gdmblue}{$\{$\epaper$\succ$\erock$\succ$\escissors$\}$}, \textcolor{gdmred}{$30\%$} are \textcolor{gdmblue}{$\{$\escissors$\succ$\epaper$\succ$\erock$\}$}, etc. We define a best response as the outcome of a voting rule on this population. For example, \emph{Borda} elects \erock{} as \textcolor{gdmblue}{our} best response.}
    \label{fig:cover}
\end{figure}
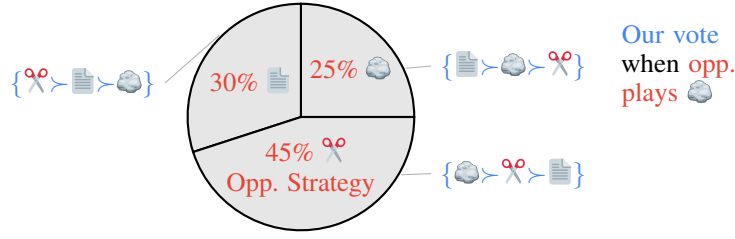

Instead of assuming access to numerical utilities, one can more generally assume each player is capable of ranking the possible outcomes~\citep{cruz2000ordinal}. Given the other players play a deterministic strategy, a rational player would simply deviate to the action that results in the outcome they most prefer. Any utility function one \emph{might} elicit through careful measurement would also achieve its maximum for that action, and so we can still consider the player to be utility-maximizing despite the lack of utilities.
Unfortunately, it is unclear how to translate this notion in the case where co-players' strategies are mixed ({\it i.e.,} randomized). Previous work has either switched from using probability theory to \emph{possibility} theory~\citep{amor2017equilibria} or introduced an additional, disinterested player to recover a mediated equilibrium~\citep{conitzer2024complexity}. Neither is able to recover
a Nash equilibrium under a traditional probabilistic framework.

A critical obstacle towards defining a Nash equilibrium in this setting is how to aggregate ordinal outcomes under mixed (probabilistic) strategies. In this work, we look towards a field that has spent centuries studying the problem of preference aggregation, a veritable ``mathematics without numbers''~\citep{kemeny1959mathwonums}, namely \emph{social choice theory}. With this viewpoint, we successfully construct a notion of a mixed Nash equilibrium called a \emph{context-ordinal} Nash equilibrium that generalizes the classical definition. %
An example of a context-ordinal equilibrium is depicted and explained in Figure~\ref{fig:cover}.

Given this new definition, many questions emerge. We establish its existence under mild conditions on aggregation (voting) rules, introduce notions of regularization, approximation, and regret, study complexity for simple settings, and develop learning rules for computing such equilibria. In doing so, we provide a generalization of Nash equilibrium that can be directly applied to elicited preferences, the \emph{fundamental data of human interactions}, which we demonstrate in two experiments: (i) general agent evaluation in Arcade, and (ii) empirical analysis of ranked-choice human leader selection (\textit{Lost at Sea}, \citep{born2022, qian2025maskmirror}).

\section{Background \& Related Work}\label{sec:related_work}

First, we review background on classical non-cooperative game theory, social choice, and prior models of equilibria assuming access to only player preferences.

\subsection{Non-Cooperative Game Theory}

A classical normal-form game (NFG) is a tuple $\langle \mathcal{N}, \mathcal{A} = (\bigtimes_{i=1}^n \mathcal{A}_i), u=(u_1, \ldots, u_n)\rangle$ where $\mathcal{N} = \{1, \ldots, n\}$ is the set of players, $\mathcal{A}_i$ is player $i$'s finite set of actions, and $u_i: \mathcal{A} \rightarrow \mathbb{R}$ is player $i$'s utility function. Players may randomize over their action sets, that is, play  \emph{mixed} strategies: $\mathcal{X}_i = \Delta^{\mathcal{A}_i}$. Their utility functions naturally extend to this domain using expected value: $u_i: \mathcal{X} = (\bigtimes_{j=1}^n \mathcal{X}_j) \rightarrow \mathbb{R} = \mathbb{E}_{\va \sim \vx}[u_i(\va)]$ where $\vx \in \mathcal{X}$ and $\va \in \mathcal{A}$. Let $x_{-i}$ denote the mixed strategy profile for players not $i$.

A Nash equilibrium (NE) is a profile $\vx$ from which no player has any incentive to deviate: $u_i(\vx) \ge u_i(z, x_{-i}) \,\, \forall \,\, i, z \in \mathcal{X}_i$.
Equivalently, each player's strategy is a best response:
\begin{align}
    x_i \in \br_i(x_{-i}) = \argmax_{z \in \mathcal{X}_i} u_i(z, x_{-i}) \,\, \forall \,\, i. \label{eqn:br}
\end{align}

\subsection{Social Choice (Voting) Theory}

Much of social choice theory studies procedures for aggregating voter preferences over alternatives such that desirable axioms are satisfied~\cite{brandt2016handbook}. The syntax $a \succ a'$ indicates a voter strictly prefers $a$ to $a'$; $a \succeq a'$, weakly prefers. A voter is indifferent between the two if $a' \succeq a$ and $a \succeq a'$, abbreviated $a \sim a'$. Each voter specifies all their preferences with a preference relation $\rho$. The set of all possible preference relations over a set of alternatives $\mathcal{A}$ is $\mathcal{P}(\mathcal{A})$.

A voting rule that determines the ``winner(s)'' (a non-empty subset, possibly with ties), is a social choice function (SCF). One that returns a ranking over the alternatives is a social welfare function (SWF). An SWF can be converted to an SCF by selecting the subset that achieves the top-rank. A probabilistic SCF (pSCF) returns a distribution (\emph{lottery}) over alternatives. An SCF can be converted to a pSCF by converting its output to a lottery with probability mass only on the winners.
In addition, we assume a voting rule can take as input a lottery over possible votes, rather than the actual set of voters' votes. For example, if the set of votes is $( \{a \succ a'\}, \{a \succ a'\}, \{a \prec a'\} )$ for three voters and two alternatives $a$ and $a'$, then a voting rule can also accept the lottery $[2/3: \{a \succ a'\}, 1/3: \{a' \succ a\} ]$. We call these \emph{doubly} probabilistic SCFs (Def.~\ref{def:dpscf}).
For intuition, we sometimes explain this concept using an infinite population of votes, each distinct vote (ballot) occurring with a given frequency. For example, Fig.~\ref{fig:cover} shows ballot $\{$\escissors$\succ$\epaper$\succ$\erock$\}$ occurring in $30\%$ of votes.

\subsection{Ordinal Games and Equilibria}

\emph{Ordinal games}~\citep{cruz2000ordinal} forgo numerical payoffs and instead assume each player can rank all joint outcomes, e.g., player $1$ of $2$ would rank $(R,S)\sim(S,P)\sim(P,R)\succ(R,R)\sim(P,P)\sim...$ in rock-paper-scissors. \citet{cruz2000ordinal} proposed a notion of approximate equilibrium for $2$-player ordinal games parameterized by \emph{order} $\{m,n\}$, indicating that player $1$ ($2$) seeks their $m$th ($n$th) ranked action when deviating. A traditional pure strategy NE in an ordinal game is equivalent to a $\{1,1\}$-NE; such an NE is not guaranteed to exist though.

\citet{conitzer2024complexity} raises the issue of defining mixed NE in ordinal games and conjectures that it ``cannot be done without access to cardinal utilities''. Instead, \citet{conitzer2024complexity} leverages the folk theorems in infinitely repeated games to construct an equilibrium that is consistent with both repeated games and mediated games. This mediated equilibrium is specified with a joint distribution $\vx$ and correlated co-player distribution $x_{-i}$ for every player $i$. These equilibria are proven to be robust in the sense that for any utility function that satisfies the ordinal constraints of the game, the pair $\vx$ and $(x_{-i})_i$ remains a mediated equilibrium.

Other work defines mixed strategies in terms of \emph{possibility} distributions rather than probability distributions~\citep{amor2017equilibria}. A possibilistic mixed strategy maps each action to an ordinal scale that can be interpreted as a preference or likelihood toward playing that action. A corresponding mixed NE can then be defined in terms of mixed possibilistic strategies.

In contrast to the mediated equilibrium, we aim to define a single factorized equilibrium profile $\vx = (x_1, \ldots, x_n)$. And in contrast to the possibilistic framework, we will define mixed strategies traditionally as probability distributions. In addition, it is unclear how any of the frameworks above might handle noisy ordinal preferences, a practically important scenario we explore later in a stochastic Condorcet election domain.

Lastly, both frameworks assume the ordinal game (OG) setting~\citep{cruz2000ordinal}, which we argue 1) is over-specified for the purposes of defining a suitable NE concept and 2) is unnatural when a player's actions are incomparable under different co-player action profiles (see Sec~\ref{sec:exp:atari}).

\section{Context-Ordinal Games \& Equilibria}\label{sec:veq}

To ascertain if a player would want to alter their strategy from a purported equilibrium, we only need to look at their possible choices, assuming their co-players' strategies remain fixed. In particular, it is \textbf{not} essential, as in an OG, to know how a player $i$ might want their co-players to change their strategies to benefit them ($i$). This leads us to the idea of a \emph{context-ordinal game} (COG), where each player ranks possible outcomes given actions chosen by everyone else. We encourage the reader to consult Fig.~\ref{fig:cover} before continuing.

\begin{definition}[Context-Ordinal Game]
A \emph{context-ordinal game} (COG) is a tuple $\langle \mathcal{N}, \mathcal{A} = (\bigtimes_{i=1}^n \mathcal{A}_i), \rho \rangle$  where $\mathcal{N} = \{1, \ldots, n\}$ is the set of players, $\mathcal{A}_i$ is player $i$'s finite set of actions, and $\pref{} = \{\pref_1, \ldots, \pref_n\}$ contains each player's conditional preference relation. Each $\pref_i: \bigtimes_{j \ne i}^n \mathcal{A}_j \rightarrow \mathcal{P}(\mathcal{A}_i)$ maps the co-players' joint action to player $i$'s preferences over $\mathcal{A}_i$.
\end{definition}

For a strategy profile to be a Nash equilibrium, each player's strategy must be a best response to the remaining players. We generalize the $\argmax$ in~\peqref{eqn:br} to mean any distribution over the player's actions that ``tie'' for top-ranked according to a social choice rule (a doubly pSCF).

\begin{definition}[Doubly Probabilistic SCF]\label{def:dpscf}
A doubly pSCF (dpSCF) is a correspondence $\nu: \Delta^{\mathcal{P}(\mathcal{A}_i)} \rightarrow 2^{\Delta^{\mathcal{A}_i}}$ from a lottery over votes to a convex set of distributions over actions.%
\end{definition}

The next definition is the key to defining our mixed Nash equilibrium concept.

\begin{definition}[Best Response with Social Choice]\label{def:br}
Player $i$'s co-players play $a_{-i}$ with probability $x_{-i}(a_{-i})$. For each $a_{-i}$ played, player $i$ conditionally specifies preferences $\rho_i(a_{-i})$ over their actions $\mathcal{A}_i$, referred to as a ``vote''\footnote{If preferences are stochastic, the vote itself may be represented as a lottery over votes.}. Let $v_i(x_{-i})$ be the resulting population of votes where each vote $\rho_i(a_{-i})$, occurs with probability $x_{-i}(a_{-i})$. Player $i$'s best response, $\br_i(x_{-i})$, is the result of the dpSCF voting rule $\nu_i$ on this population of votes.%
\end{definition}

The following NE definition is standard. Our primary innovation is Def.~\ref{def:br} of a best response.
\begin{definition}[Context-Ordinal Nash Equilibrium]
A strategy profile $\vx$ is a context-ordinal Nash equilibrium (CO-NE) iff $x_i \in \br_i(x_{-i})$ for all i.
\end{definition}
In words, $\vx$ is an NE if every player's mixed strategy only places mass on winning candidate actions.
See~\appxref{app:corr} for a definition of a correlated equilibrium.

Our definition is naturally robust to mis-specification of utilities (assuming they exist), a key focus of~\citep{conitzer2024complexity}. The underlying utilities are allowed to change as long as the partial ranking of actions does not. Because our definition only observes the partial ranking, the NE is invariant to these changes in utility.

\subsection{Existence of Mixed Context-Ordinal NE}\label{sec:veq:existence}

The social choice best response operator $\br_i$ for each player $i$ maps from a partial profile $x_{-i}$ to a non-empty, convex subset of the simplex. Denote upper hemicontinuous by u.h.c.

\begin{theorem}[\citeauthor{kakutani1941generalization}~\citeyear{kakutani1941generalization}]\label{thm:exist}
If each $\br_i$ is u.h.c., then an NE exists in mixed strategies.
\end{theorem}

A set-valued mapping $\br_i$ is u.h.c. if for every convergent sequence of co-player strategies $\{\textcolor{gdmred}{x_{-i}^t}\}_{t}$, e.g., the distribution over votes illustrated in Fig.~\ref{fig:cover}, and for any convergent sequence of player $i$'s best responses $\{\textcolor{gdmblue}{x_i^t}\}_{t}$ with $\textcolor{gdmblue}{x_i^t} \in \br_i(\textcolor{gdmred}{x_{-i}^t})$, the $\lim_{t \rightarrow \infty} \textcolor{gdmblue}{x_i^t}$ lies in $\br_i(\lim_{t \rightarrow \infty} \textcolor{gdmred}{x_{-i}^t})$.

Traditional SCFs may exhibit discontinuities in their elected candidates as the population of votes slightly changes. By considering the distributions returned by dpSCFs, we can more naturally study their u.h.c. properties.
The u.h.c. definition handles subtleties that arise, for example, when the distribution of votes changes from one that prefers candidate A to one in which candidate A and B are tied. In the latter case, any distribution over A and B is \uline{valid} and samples a suitable winning candidate: $\br_i(\lim_{t \rightarrow \infty} \textcolor{gdmred}{x_{-i}^t}) = [z, 1-z], z \in [0,1]$. And if the limit of the winning candidate distribution specifically selected A ($\lim_{t \rightarrow \infty} \textcolor{gdmblue}{x_i^t} = [1, 0]$), that would still satisfy the u.h.c. condition because A is in the set of \uline{valid} candidate distributions.

We identify several u.h.c. families of voting rules:
a) score voting where each vote assigns candidates numerical scores and the candidate with highest ($x_{-i}$-weighted) score wins, b)
positional voting~\citep{saari1995basic}, e.g., Borda counts, c)
probabilistic voting~\citep{brandl2016consistent}, e.g., maximal lotteries~\citep[p. 30]{fudenberg1991game}, d)
and social grading functions~\citep{balinskilaraki2007measuring}.
Scoring and positional voting rules induce normal-form games. Classical NE implicitly uses score voting (see~\appxref{app:cog:nfg}) where voters score candidate actions with their precise payoffs.
Later we introduce regularized best responses, rendering any dpSCF voting rule u.h.c. and ensuring existence of their NE.

\paragraph{Complexity of CO-NE}
In~\appxref{app:complexity}, we study complexity of CO-NE and show that there exist intuitively adversarial $2$-player games that when studied under simple voting rules (e.g., Borda counts) map to normal-form games that are not zero-sum. The implication is that CO-NE are not polynomial-time computable. Whereas prior work sought to define equilibria of ordinal games that lie in $\textsf{P}$ at the expense of a departure from traditional representation, our aim is to define an equilibrium notion that generalizes Nash: probabilistic, factorizes, and gracefully reduces to classical NE under assumptions, e.g., score voting where score$=$utility.

\section{Learning and Approximation}

Gradient descent serves as the workhorse of learning in games. Its interpretation as a proximal operator is important here because it allows us to view (projected) gradient descent as the solution to a regularized optimization problem. This view appears in related algorithms like follow the regularized leader~\citep{mcmahan2017survey,suggala2020follow} and mirror descent as well~\cite{beck2003mirror}:
\begin{align}
x_i' = \Pi_{\mathcal{X}_i}[x_i + \eta \nabla_{x_i} u_i(x_i, x_{-i})] &= \argmax_{z \in \mathcal{X}_i} u_i(z, x_{-i}) - \frac{1}{2\eta} ||z - x_i||^2
\end{align}
where $x_i'$ denotes the next iterate, $\Pi_{\mathcal{X}_i}$ denotes the Euclidean projection onto the set $\mathcal{X}_{i}$, and $\eta$ is a step size parameter (equiv., inverse regularization coefficient). Convergence of gradient descent is typically analyzed in terms of the successive distance between iterates, $||x_i' - x_i||$. Notice that in unconstrained settings when the projection operator acts as an identity, this simplifies to $||\eta \nabla_{x_i} u_i(x_i, x_{-i})||$ and is hence proportional to the norm of the gradient. Gradient norms are used as both performance metrics and constructing loss functions to develop other algorithms~\citep{gempapproximating}. It should not come as a surprise then that our first task is to replicate a technique to regularize our best response definition. Doing so will provide us with methods for learning as well as metrics to measure performance of those learning algorithms.

\subsection{Regularization of Best Responses}\label{sec:reg}

\newcommand{\coinH}{\tikz[baseline=-0.7ex]{\fill[black] (-0.1,0) circle (0.55em); \node[font=\scriptsize, text=white, inner sep=0pt] at (0,0) {\textsf{H}};}}
\newcommand{\coinT}{\tikz[baseline=-0.7ex]{\fill[white] (0,0) circle (0.55em); \draw[thick, black] (0,0) circle (0.55em); \node[font=\scriptsize, text=black, inner sep=0pt] at (0,0) {\textsf{T}};}}%
\begin{figure}[t]
    \centering
    \begin{tikzpicture}[scale=1.5, >=Stealth]
    \node[font=\footnotesize, anchor=east, inner sep=4pt] (steps) at (-2.0, 0) {%
        \begin{tabular}{@{}r@{\;}l@{}}
            \textcolor{black}{\textbf{1.}}
                & %
                  $\textcolor{maroon}{\hat{x}_{-i}} \sim
                   \mathrm{Dir}\!\big(\mathbf{1}
                   + \tfrac{1}{q} \textcolor{gdmred}{x_{-i}}\big)$
                \\[6pt]
            \textcolor{black}{\textbf{2.}}
                & \erock\textcolor{navyblue}{%
                  $\;\sim \mathrm{Cat}(\mu_i)$}
                \\[6pt]
            \textbf{3.}
                & \coinH$\;\sim \mathrm{Bern}(p)$
        \end{tabular}%
    };
    \begin{scope}[shift={(1.8, 0)}]
        \draw[thick, fill=black!10] (0,0) -- (0:1) arc (0:79.2:1) -- cycle;
        \draw[thick, fill=black!10] (0,0) -- (79.2:1) arc (79.2:190.8:1) -- cycle;
        \draw[thick, fill=black!10] (0,0) -- (190.8:1) arc (190.8:360:1) -- cycle;
        \node[text=maroon] at (25:0.55) {22\% \erock};
        \node[text=maroon] at (135:0.55) {31\% \epaper};
        \node[text=maroon] at (275.4:0.35) {47\% \escissors};
        \node[anchor=west] (ballot1) at (1.15, 0.45) {%
            \textcolor{gdmblue}{$\{$\epaper$\succ$\erock$\succ$\escissors$\}$}%
        };
        \draw[thin, black!30] (25:1.02) -- (ballot1.west);
        \node[text=black!40, anchor=east] (origballot) at (-1.2, 0.1) {%
            \textcolor{gdmblue}{$\{$\escissors$\succ$\epaper$\succ$\erock$\}$}%
        };
        \draw[thick, black!40] (origballot.west) -- (origballot.east);
        \node[text=navyblue, above=2pt of origballot] (usurpballot) {%
            $v_u{=}\{$\erock$\succ$\epaper$\sim$\escissors$\}$%
        };
        \draw[thin, black!30] (135:1.02) -- (origballot.east);
        \node[anchor=west] (ballot3) at (1.15, -0.5) {%
            \textcolor{gdmblue}{$\{$\erock$\succ$\escissors$\succ$\epaper$\}$}%
        };
        \draw[thin, black!30] (330:1.02) -- (ballot3.west);
        \fill[white] (60:0.88) circle (0.08);
        \draw[thick, black] (60:0.88) circle (0.08);
        \node[font=\tiny, text=black] at (60:0.88) {\textsf{T}};
        \fill[black] (170:0.88) circle (0.08);
        \node[font=\tiny, text=white] at (170:0.88) {\textsf{H}};
        \fill[white] (240:0.88) circle (0.08);
        \draw[thick, black] (240:0.88) circle (0.08);
        \node[font=\tiny, text=black] at (240:0.88) {\textsf{T}};
    \end{scope}
    \end{tikzpicture}
    \caption{%
        Algorithm~\ref{alg:reg_br} applied to the example from
        Fig.~\ref{fig:cover}.
        \textbf{Step~1:} Original co-player strategy
        \textcolor{gdmred}{$x_{-i} = [25\%, 30\%, 45\%]$} is perturbed via
        \textcolor{maroon}{$\hat{x}_{-i}$} $\sim
        \mathrm{Dir}(\mathbf{1} + \frac{1}{q}$ \textcolor{gdmred}{$x_{-i}$}$)$,
        yielding \textcolor{maroon}{$[22\%, 31\%, 47\%]$};
        as $q \to 0$, $\textcolor{maroon}{\hat{x}_{-i}} \to \textcolor{gdmred}{x_{-i}}$.
        \textbf{Step~2:} A usurper action
        $u{=}$\erock{} is sampled from
        \textcolor{navyblue}{$\mathrm{Cat}(\mu_i)$}, giving
        \textcolor{navyblue}{$v_u = \{$\erock$\succ$\epaper$\sim$\escissors$\}$}
        with \erock{} top-ranked and all others tied.
        \textbf{Step~3:} Each vote is independently replaced
        by \textcolor{navyblue}{$v_u$} with probability $p$ (coin flip).
        When $p=0$ the population is unchanged; when $p=1$
        all votes become~\textcolor{navyblue}{$v_u$}. A voting rule $\nu_i$ determines a best response from this perturbed population. Results are averaged over trials to obtain the final \emph{regularized} best response.
    }
    \label{fig:perturbation}
\end{figure}
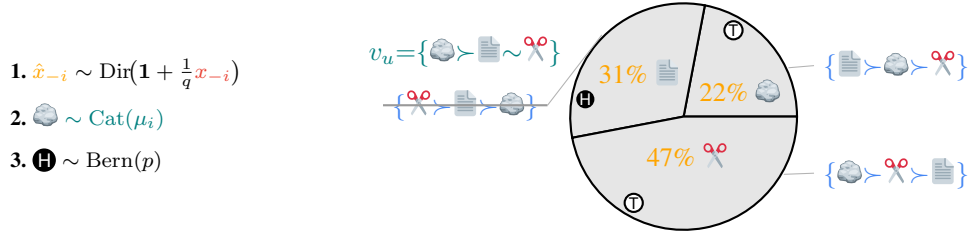

Regularization is a useful tool in game theory for selecting equilibria~\citep{mckelvey1995quantal,harsanyi1988general}, aiding convergence~\citep{perolat2021poincare}, imitating target play~\citep{gempconvex,bakhtinmastering}, and online learning (FTRL)~\citep{shalev2012online}. An exemplar is KL-divergence: $\textsf{KL}(\target \parallel x) = \sum_{\mathcal{A}} \target(a) \log(\target(a) / x(a))$, which measures how much an approximate distribution $x$ differs from a true distribution $\target$. The reverse KL, $\textsf{KL}(x \parallel \target)$ was used in~\citep{bakhtin2022mastering} to regularize learned strategies in Diplomacy towards recorded human play.

Given that COGs lack utilities, making how to achieve direct regularization unclear, we take the approach of regularizing via random perturbation.
Let $BR_i^{(p,q,\mu_i)}$ denote the regularized best response operator parameterized by target distribution $\mu_i \in \Delta^{\mathcal{A}_i}$, perturbation hyperparameter $q \ll 1$, and regularization strength $p \in [0, 1]$. We set $\mu_i = \sfrac{1}{\vert \mathcal{A}_i \vert}\mathbf{1}$ always.

A regularized best-response in COGs should satisfy the following desiderata for every $x_{-i}$:
\begin{enumerate}[label=(\textbf{C\arabic*}), ref=C\arabic*]
    \item \label{des:target} Maximum regularization ($p=1$) results in $\target_i$ for any target distribution $\target_i$;%
    \item \label{des:no_effect} Zero regularization ($p=0$) returns an element of $\br_i$ from Def.~\ref{def:br};
    \item \label{des:unique} Any regularization results in a single-valued best response;
    \item \label{des:cont} The regularized best response is a continuous function of $x_{-i}$.
\end{enumerate}

For clarity, we use a concrete example in Fig.~\ref{fig:perturbation} to describe our approach which satisfies the above conditions and defer Algorithm~\ref{alg:reg_br} and a more rigorous discussion to~\appxref{app:reg_br}.

Next, we will leverage regularization to construct notions of approximation and regret in COGs. Later in Section~\ref{sec:exp}, we also use it to construct algorithms, e.g., FTRL, to approximate context-ordinal Nash equilibria in experiments. \appxref{sec:alg} reviews learning algorithms.

\subsection{Performance Metrics}\label{sec:approx}

Approximate Nash equilibria are most often judged on how much any player can gain by deviating, referred to as \emph{exploitability}, $\epsilon = \max_i \epsilon_i$, where:
\begin{align}
    \epsilon_i(\vx) &= \max_{z \in \Delta^{\mathcal{A}_i}} u_i(z, x_{-i}) - u_i(\vx), \label{eqn:classical_eps_i}
\end{align}
and $\epsilon_i$ is sometimes referred to as \emph{immediate} regret. That is because of the tight connection between game theory and online learning~\citep{gordon2008no}. (External) regret for a sequence of strategies $[x_{i,t}]_{t \in [1, T]}$ measures exploitability over $T$ rounds of play: $\max_{z \in \Delta^{\mathcal{A}_i}} \sum_{t=1}^T u_{i}(z, x_{-i, t}) - u_{i}(\vx_t)$.
COGs do not provide utilities so we explore alternative notions of approximation and regret.

\subsubsection{Strategy Space}\label{sec:approx:strat}

We can measure a distance to NE in strategy space as $\epsilon = \max_i \epsilon_i$,
\begin{align}
    \epsilon_i(x) &= \min_{z \in \br_i(x_{-i})} D(z, x_i) \label{eqn:exp_strat_noreg}
\end{align}
and $D$ is continuous and non-negative. For example, $D$ could be earth mover's distance which results in $\epsilon_i$ summing the amount of probability mass that player $i$ has placed on strictly losing candidates.

As defined, this function may be discontinuous because of jumps in the u.h.c. best response set $\br_i$. We can replace the feasible set by the \emph{regularized} best response $\br_i^{(p,q,\target_i)}$ which is continuous. By Berge's maximum theorem~\citep[Theorem 17.31]{aliprantis2006},
\begin{align}
    \epsilon_i^{(p,q,\target_i)}(\vx) &= \min_{z_i \in \br_i^{(p,q,\target_i)}(x_{-i})} D(z_i, x_i) \label{eqn:emd}
\end{align}
is continuous in $\vx$.
Figure~\ref{fig:nwon_chicken} in~\appxref{app:approx} illustrates this measure for the classic Chicken game. While this measure is suitable for immediate regret, it is unclear how to extend it to regret which evaluates a sequence of strategies.

\begin{figure*}
    \centering
    \begin{subfigure}[t]{0.48\textwidth}
        \centering %
        \caption{Exploitability} %
        \label{fig:eps} %
        \includegraphics[width=1.0\linewidth]{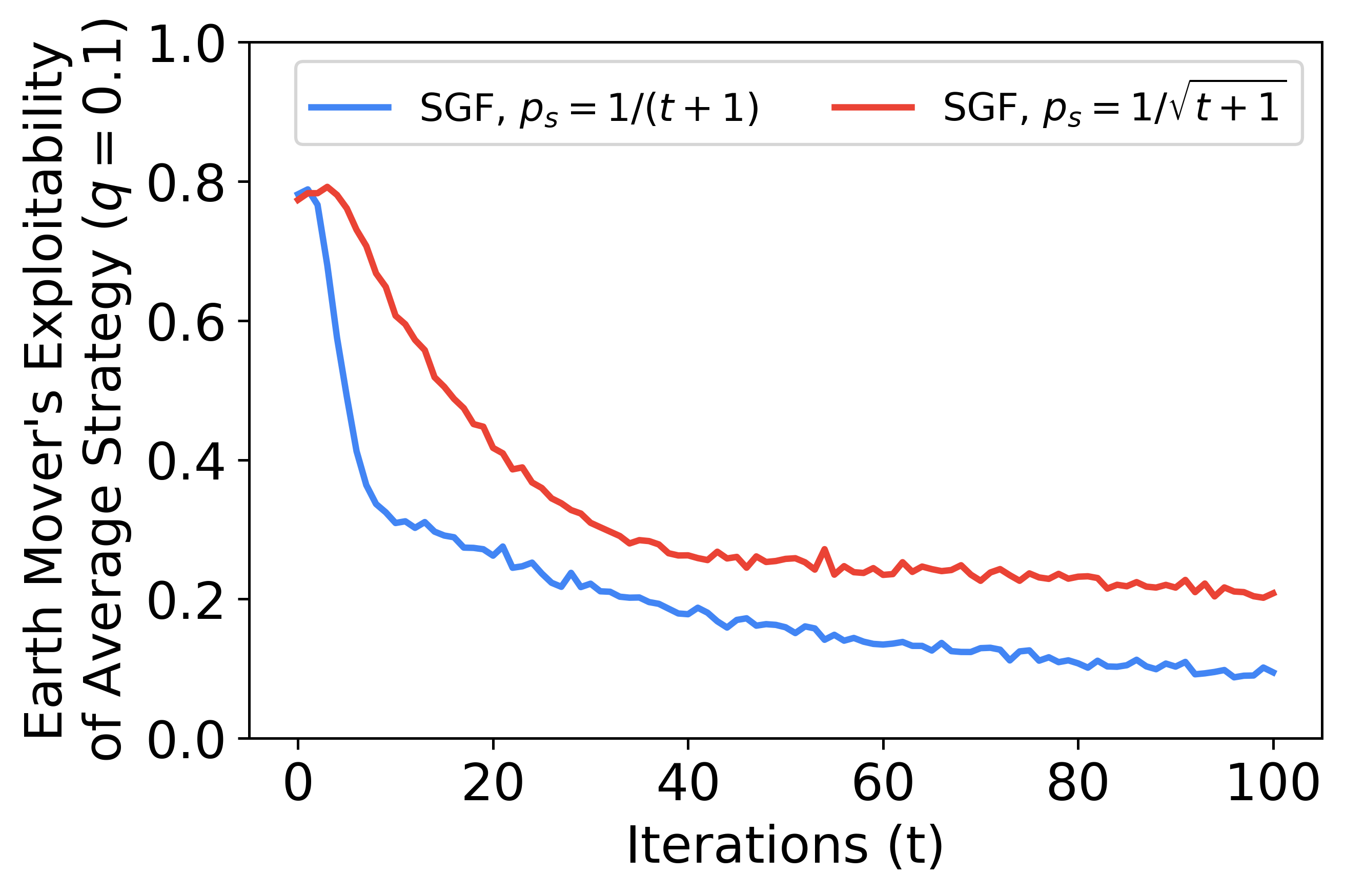}
    \end{subfigure}
    \hfill
    \begin{subfigure}[t]{0.48\textwidth}
        \centering %
        \caption{Hindsight Winrate} %
        \label{fig:regret} %
        \includegraphics[width=1.0\linewidth]{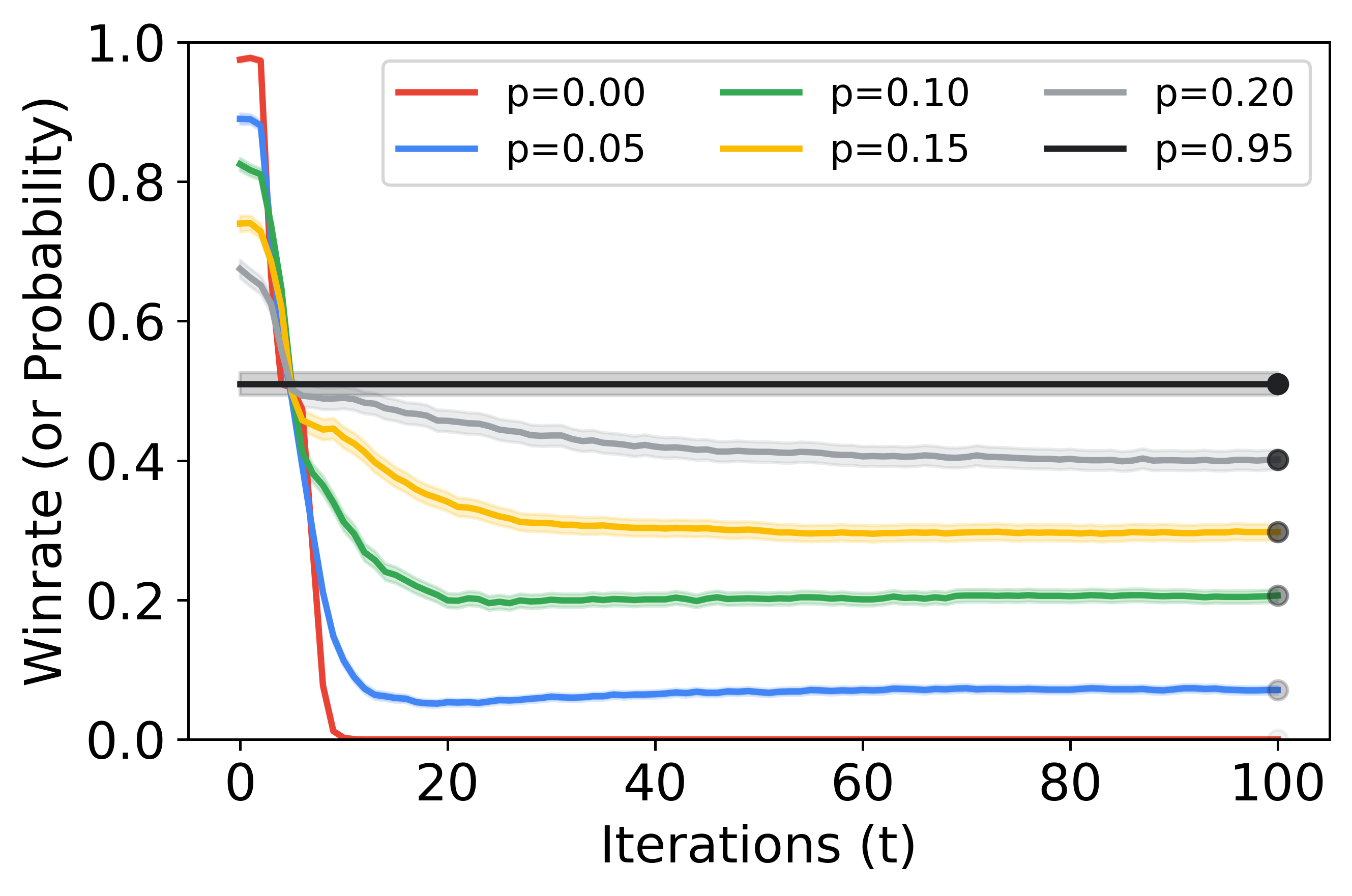}
    \end{subfigure}
    \caption{We evaluate our (SGF-based) FTRL approach ($p_s$ indicates the \emph{solver}'s $\br_i^{(p_s,q,\mu_i)}$ parameter) on the Atari evaluation game according to the metrics in Section~\ref{sec:approx}. (\subref{fig:eps}) EMD as defined in Section~\ref{sec:approx:strat} with $p=0$, eqn.~\peqref{eqn:emd}; (\subref{fig:regret}) Hindsight winrate from Section~\ref{sec:approx:winrate}. Panel~(\subref{fig:regret}) sets $p_s=\sfrac{1}{t+1}$. FTRL can also be considered a \emph{smoothed}-FP approach~\citep{fudenberg1995consistency}.}
    \label{fig:atari_regret}
\end{figure*}

\subsubsection{Meta-Game Analysis}\label{sec:approx:winrate}

Another choice directly asks ``would you prefer to have played a fixed $z_i^*$ in hindsight?'' To define the best fixed $z_i^*$ in hindsight, simply collect all the ``votes'' generated at each round $t$ by the co-players in proportion to their strategies $x_{-i,t}$. Use the (dpSCF) voting rule to aggregate the votes across all rounds and return the best probabilistic strategy $z_i^*$.

To measure the regret, consider the same sequence of co-player strategies as before, but apply the voting rule, only considering the hindsight and online actions in proportion to their appearance in their mixed strategies. For example, if at $t=t'$, $z_i^* = [1, 0]$ and $x_{i,t'} = [0.5, 0.5]$, in proportion to each $a_{-i} \sim x_{-i}$ played at time $t'$, generate two votes with equal representation: the first action compared against itself and the first action against the second. As before, collect all of the generated votes and apply the voting rule (equiv. best response operator $\br_i$) to determine the winner.
We present the probability of hindsight being selected by $\br_i^{(p,q=0.1,\mu_i=[\sfrac{1}{2},\sfrac{1}{2}])}$ in Figure~\ref{fig:regret} over iterations.

In~\appxref{sec:approx:mov}, we explore one more additional metric derived from the social choice perspective, namely \emph{margin of victory}, which measures the proportion of votes ($x_{-i}$) that must be altered for a given strategy $x_i$ to become a best response $x_i \in \br_i(x_{-i})$.

\subsection{Approximating Classical Nash Equilibria}

The focus of this work is on games with only revealed ordinal preferences. Nevertheless, one might still hope to approximate a classical NE of the game that is defined by \emph{hidden} cardinal utilities. We obtain non-trivial approximation bounds for that setting by appealing to results from the study of \emph{distortion} within social choice. Let $u_i \rhd \rho_i$ mean any utility $u_i$ consistent with preferences $\rho_i$. Then the (additive $+$) distortion of a voting rule $\nu_i$,
\begin{align}
    d^+(\nu_i) &= \max_{\rho_i, \vx} \max_{u_i \rhd \rho_i} [ u_i(\vx) - u_i(\br_i, x_{-i})], \label{eqn:add_distortion}
\end{align}
captures the cost of electing candidates using ordinal information instead of voters' more nuanced cardinal utilities (ordinal information is assumed consistent with utilities).

It is typically assumed that each voter attributes cardinal utilities to candidates such that they are non-negative and sum-to-$1$, ``one person, one vote''~\citep{procaccia2006distortion}. This property does not generally hold in games\textemdash a player's payoffs do not sum-to-$1$ under each possible co-player action $a_{-i}$. Nevertheless, equilibria are invariant to affine transformations of payoffs, so we can shift and scale each player's payoffs such that they are all non-negative with positive sum. This still leaves varying payoff-sums ($\vs$), but we can recover a bound on the suboptimality of the best responses computed using voting rules versus expected payoff maximization.
\begin{theorem}\label{lemma:scaled_distortion}
Let $\nu_i$ be a dpSCF voting rule and $\br_i$ its induced best response (Def.~\ref{def:br}). Let $d^+(\nu_i, u_i, x_{-i}) = \max_{z} u_i(z, x_{-i}) - u_i(\br_i, x_{-i}) \ge 0$ denote the suboptimality of the best response to $x_{-i}$ computed using $\nu_i$. Finally, let $\vs(u_i) = \{\sum_{a_i \in \mathcal{A}_i} u_i(a_i, a_{-i})\}_{a_{-i} \in \mathcal{A}_{-i}}$ denote the set of player $i$'s payoff sums under each co-player action profile. Then
\begin{align}
d^+(\nu_i, u_i, x_{-i}) &\le \bar{d}^+(\nu_i, \vs(u_i)) = \kappa^+ + (\min_k s_k) d^+(\nu_i)
\end{align}
where $\kappa^+ = (\max_k s_k - \min_k s_k)$ and $d^+(\nu_i)$ is the \emph{additive} distortion $\nu_i$ from~\peqref{eqn:add_distortion}.
\end{theorem}

Theorem~\ref{thm:reg_distortion} shows our regularized best response $\br_i^{(p,q,\mu_i)}$ only introduces an additional distortion to $d^+(\nu_i)$ that is at most linear in $p$ for small $p$.
We similarly derive a multiplicative bound $d(\nu_i, u_i, x_{-i}) = \max_{z} \frac{u_i(z, x_{-i})}{u_i(\br_i, x_{-i})} \le \kappa d(\nu_i)$ where $\kappa = \frac{\max_k s_k}{\min_k s_k}$.
There exist voting rules $\nu_i$ that guarantee additive distortion of $\sfrac{1}{4} \le d^+(\nu_i) \le \sfrac{1}{2} (1 - \sfrac{1}{\vert \mathcal{A}_i \vert^2})$ and achieve $\tilde{\mathcal{O}}(\sqrt{\vert \mathcal{A}_i \vert})$ for multiplicative distortion $d(\nu_i)$~\citep{caragiannis2017subset}.

One of the earliest decentralized learning algorithms, (weakened) fictitious-play (WFP), iterates by approximately best responding to the co-player's historical play. The upper bound on the approximation error just derived in Theorem~\ref{lemma:scaled_distortion} can be used to directly derive bounds on the approximation error of an equilibrium computed using WFP.

\begin{corollary}[Theorem 1~\citep{conitzer2009approximation}]\label{cor:cofp}
The WFP profile $\vx_T = [x^{(1)}_T, x^{(2)}_T]$ after $T$ rounds is an $\epsilon_T$-NE where $\epsilon_T \le \frac{T+1}{2T} + \frac{T-1}{2T} \max_i \bar{d}^+(\nu, \vs(u_i))$.
\end{corollary}

Shift and scale the payoffs w.l.o.g. so that $\max_k s_k = 1$. Theorem~\ref{lemma:scaled_distortion} combined with Corollary~\ref{cor:cofp} and the known distortion bounds listed above imply there exists a voting rule $\nu$ such that FP converges to, at worst, a $(1 - \sfrac{1}{4} \min_k s_k)$-NE in $2$-player general-sum. Note that $\sfrac{1}{2}$ is the known lower bound for approximating NE with constant size support~\citep{feder2007approximating}. With only access to ordinal feedback, FP achieves $\sfrac{3}{4}$ for $\min_k s_k = 1$, a loss of $\sfrac{1}{4}$ in expected payoff.

Recent work provides algorithms for approximating coarse correlated equilibria assuming player's rankings over actions are drawn according to a Placket-Luce model consistent with true underlying utilities~\citep{liu2026online}. This allows them to estimate the utilities by observing sampled rankings, i.e., invert the model. In contrast, we do not assume it is possible to discover the underlying utilities. For example, if a user deterministically reveals their preferences to be a single ranking, we can never learn the gaps in utility between actions. Our aim in this work is to develop a theory that natively handles ordinal information without (implicitly) working with any presumed underlying cardinal information.

\begin{figure*}
    \centering
    \begin{subfigure}[t]{0.32\textwidth}
        \centering %
        \caption{Agent Marginals} %
        \label{fig:agent_marginals} %
        \includegraphics[width=1.0\linewidth]{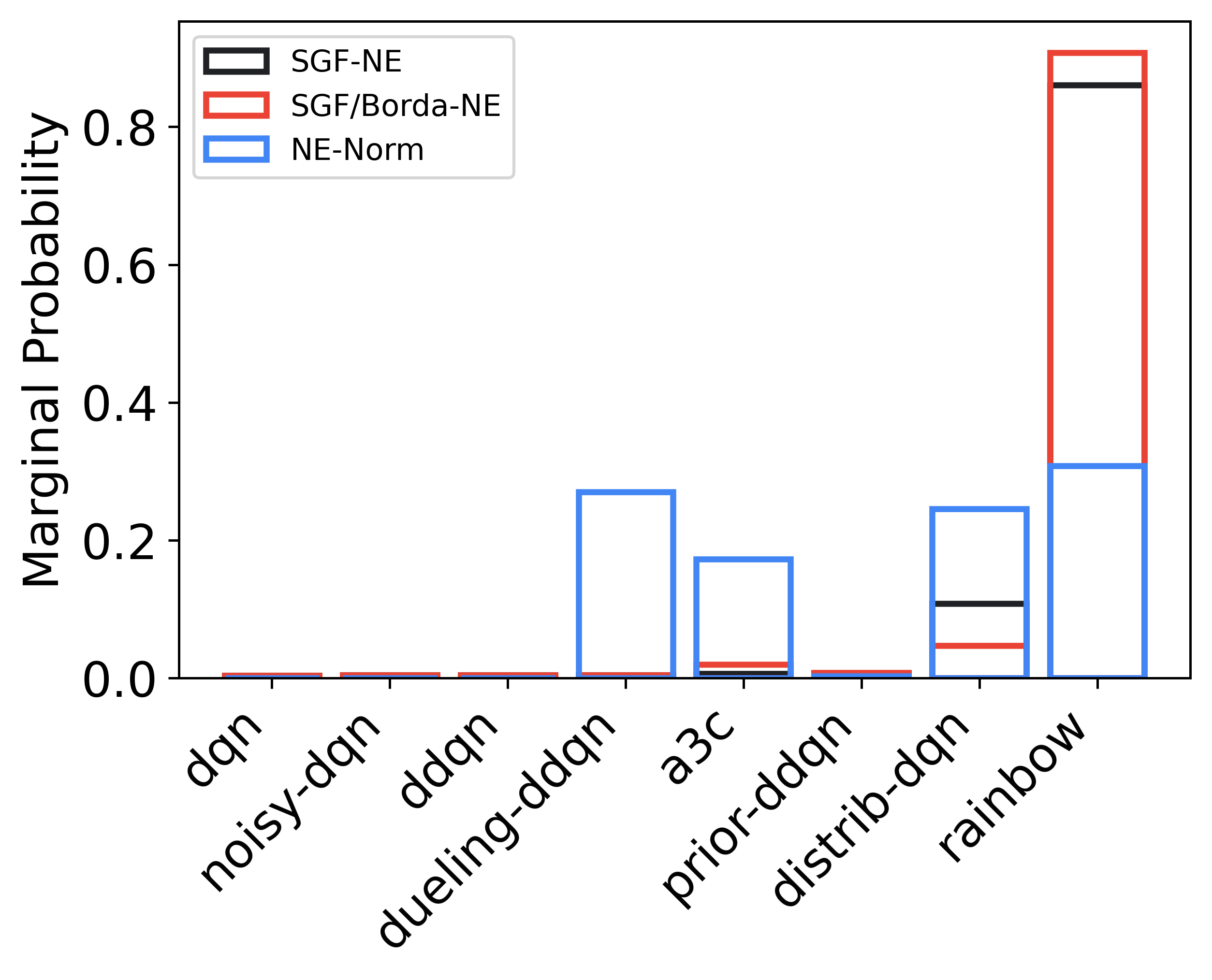}
    \end{subfigure}
    \hfill
    \begin{subfigure}[t]{0.32\textwidth}
        \centering %
        \caption{Game Marginals} %
        \label{fig:game_marginals} %
        \includegraphics[width=1.0\linewidth]{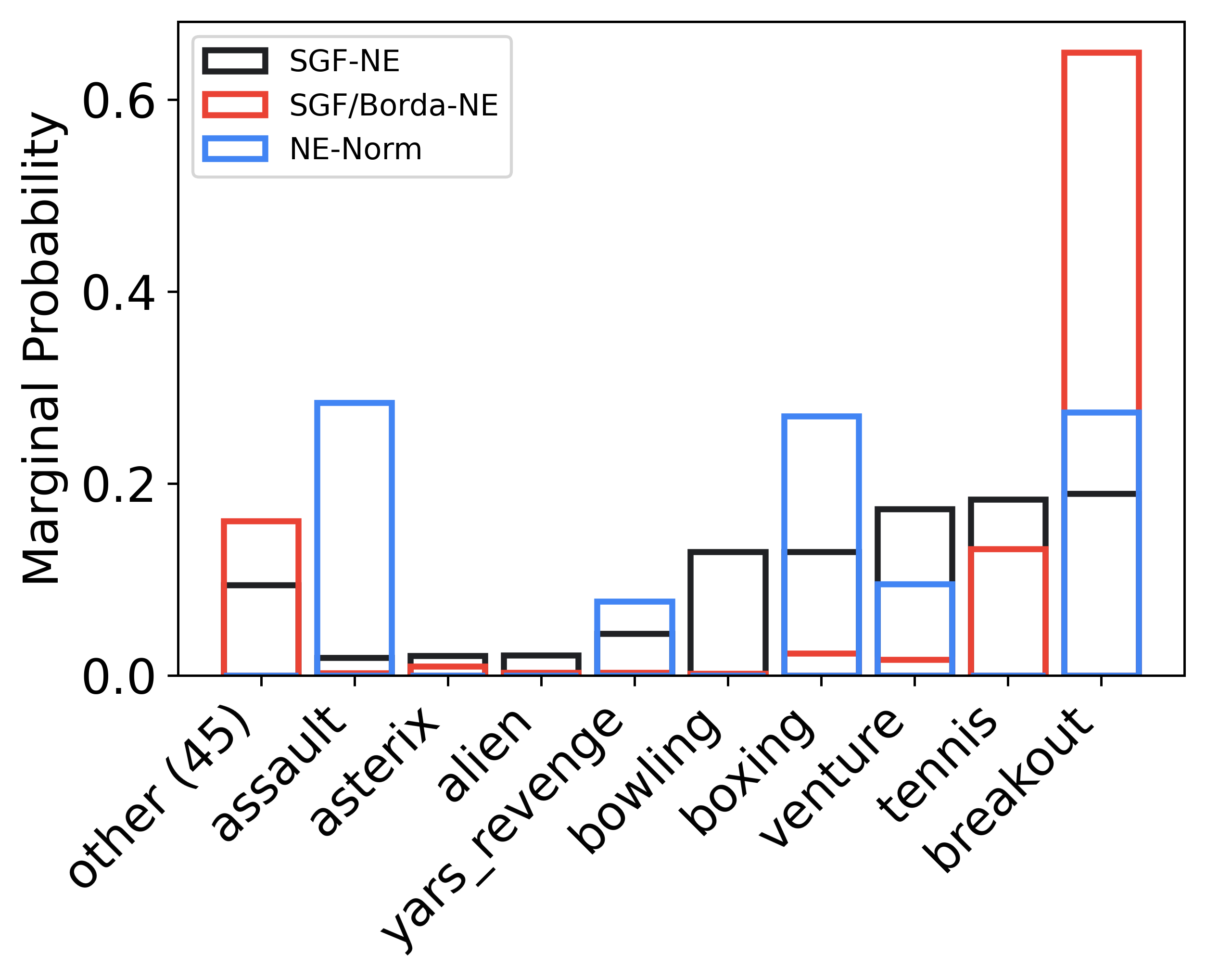}
    \end{subfigure}
    \hfill
    \begin{subfigure}[t]{0.32\textwidth}
        \centering %
        \caption{Agent Marginals Diffs} %
        \label{fig:agent_marginals_conv} %
        \includegraphics[width=1.0\linewidth]{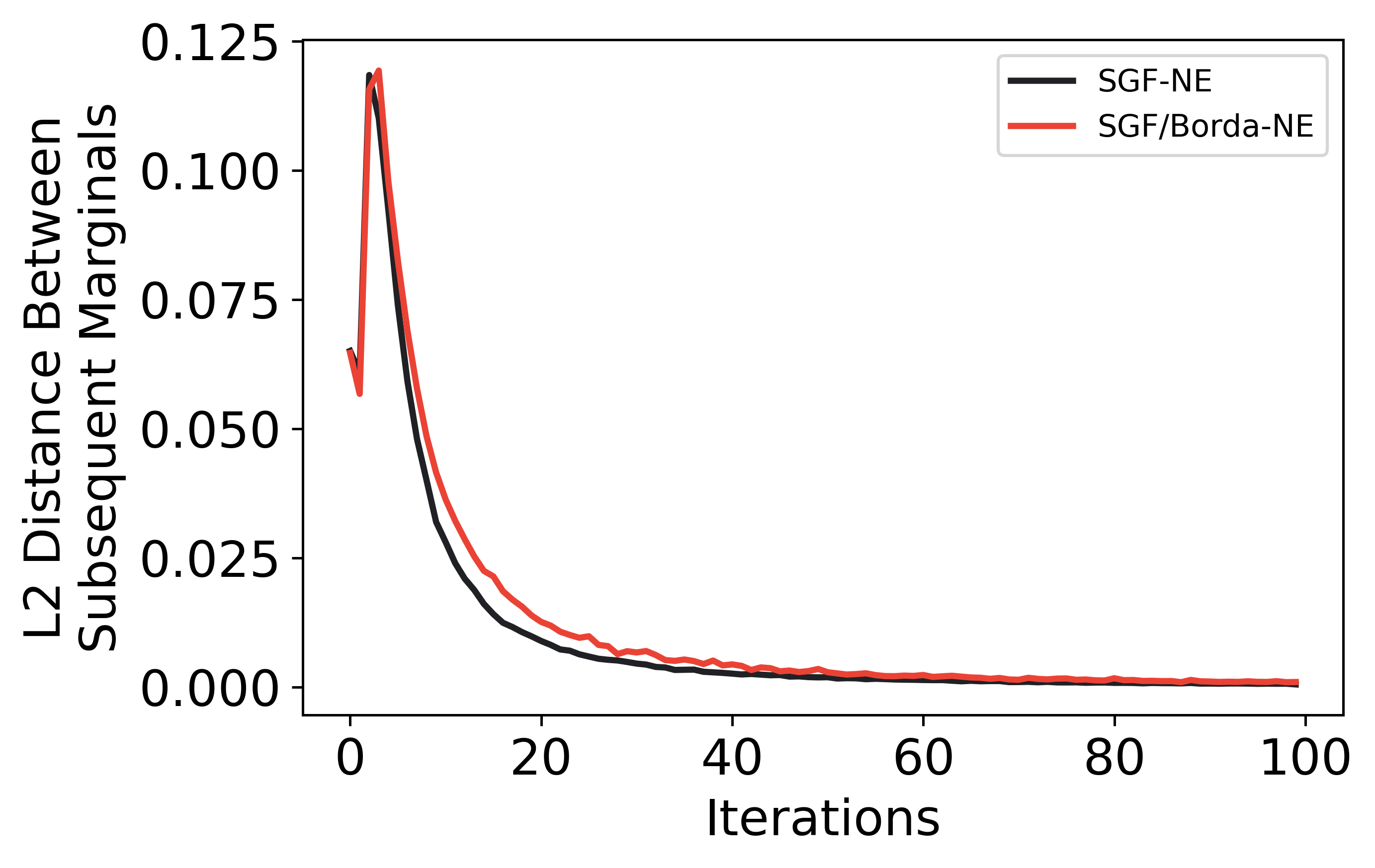}
    \end{subfigure}
    \caption{Atari: We present two different CO-NEs computed using the FTRL-inspired approach: (SGF-NE) both agents use social grading functions and (SGF/Borda-NE) where the task agent instead uses Borda. We compare them to the NE of an agent vs task performance matrix with scores normalized to $[0, 1]$ as in~\citep{balduzzi2018re}. Figure~\ref{fig:agent_marginals_conv} supports convergence of our FTRL-inspired solver. Figure~\ref{fig:atari_regret} displays additional performance information for the SGF-NE approach.}
    \label{fig:atari_results}
\end{figure*}

\section{Applications}\label{sec:exp}

We demonstrate our theory in two different domains: a) game-theoretic evaluation of AI agents and b) a tactical voting setting with stochastic election outcomes and preferences.

\subsection{General Agent Evaluation in the Arcade Learning Environment}\label{sec:exp:atari}

We consider a game-theoretic~\citep{balduzzi2018re,liure,marris2025deviation} evaluation experiment with Atari agent vs task performance data~\citep[Table 5]{hessel2018rainbow} (c.f. Figure~\ref{fig:adv_game}). In this setting, the performance matrix $A$ plays the role of the payoff matrix in a zero-sum bi-matrix game: $\min_{x} \max_{y} x^\top A y$ where $x$ ($y$) is the distribution over agents (tasks). The solution surfaces a distribution over agents that is robust against an adversarially chosen task distribution (NE-Norm in Figure~\ref{fig:atari_results}). We instead run a version of FTRL (FTPL)~\citep{mcmahan2017survey,suggala2020follow} using our proposed regularized best response. We use social grading functions (SGF-NE) as the voting rule and assign each agent one of four grades on a task based on quartiles. Figure~\ref{fig:atari_results} also considers a heterogeneous setting where the agent player uses SGF and task player uses Borda (SGF/Borda-NE).

Figure~\ref{fig:atari_regret} presents performance of our FTRL-style solver according to the metrics discussed in Sections~\ref{sec:approx}: exploitability~\peqref{eqn:exp_strat_noreg} and hindsight winrate. All indicate our algorithm is effectively learning. Figure~\ref{fig:atari_results} displays the learned equilibrium for the agent and task player along with a plot supporting convergence of the iterates.
Prior work espousing voting-as-evaluation found the rainbow agent top-ranked according to $9$ different popular voting rules~\citep[Table 2]{lanctot2023evaluating}. Their result can be interpreted as equivalent to ours if we constrained the distribution over games to be uniform. Instead, our theory uniquely allows one to show that rainbow is the strongest agent even if the distribution over games is chosen adversarially.

\subsection{Voting Equilibrium from Human Election Data: Lost at Sea}

In the \textit{Lost at Sea} election scenario~\citep{born2022}, groups of four participants deliberate and then elect one member to complete a \emph{leader's task}, a quiz whose score determines the payout to the group.
The election proceeds as follows. Each participant submits both a \emph{willingness-to-lead} self-nomination score (\wtl{} $\in \{0, \ldots, 10\}$) and a vote ranking the other participants.
The election mechanism selects the two participants with highest \wtl{}, splitting ties randomly, and executes a runoff between them.
Note that because no one could vote for themselves, only the votes of participants that are not candidates impact the result; see more details on the election in~\appxref{app:las_instructions}.

\begin{table}[t!]
    \centering
    \caption{Election: (Pure) Strategies are \emph{not} a CO-NE.}
    \label{tab:not_pure_ne_group}
    \begin{tabular}{c|c|c|c} \toprule
        \emph{player} $i$ & \multicolumn{2}{|c|}{\emph{combinatorial action} $a$} & \emph{implicitly def.} $\rho_i(a_{-i})$ \\
        voter & vote & \wtl{} & \texttt{pref} \\ \midrule
        \edolphin{} & \eelephant{} $\succ$ \eturtle{} $\succ$ \ezebra{} & 2 & \eelephant{} $\succ$ \eturtle{} $\succ$ \ezebra{} $\succ$ \edolphin{} \\
        \eelephant{} & \ezebra{} $\succ$ \eturtle{} $\succ$ \edolphin{} & 9 & \ezebra{} $\succ$ \eturtle{} $\succ$ \eelephant{} $\succ$ \edolphin{} \\
        \eturtle{} & \eelephant{} $\succ$ \ezebra{} $\succ$ \edolphin{} & 5 & \eturtle{} $\succ$ \eelephant{} $\succ$ \ezebra{} $\succ$ \edolphin{} \\
        \ezebra{} & \eelephant{} $\succ$ \eturtle{} $\succ$ \edolphin{} & 10 & \ezebra{} $\succ$ \eelephant{} $\succ$ \edolphin{} $\succ$ \eturtle{} \\ \bottomrule
    \end{tabular}
\end{table}

Separately, participants submitted rankings (\texttt{pref}) over who is elected. Given fixed co-player actions $a_{-i}$ and the election outcomes $o_i$ resulting from each possible action $a_i \in \mathcal{A}_i$, we can rank the actions by \texttt{pref}($o_i$), e.g., if for \edolphin{}, $a \rightarrow o =$ \eelephant{} and $a' \rightarrow o' =$ \eturtle{}, then $a \succ a'$.

Voters can manipulate the election by pulling themselves out of the race (lowering \wtl{}). Alternatively, if they know that the other players prefer a different candidate, they can reorder their vote to put their second favorite higher, therefore increasing the odds of that candidate being elected.

Therefore, each participant's strategy space ($\vert \mathcal{A}_i \vert = 66$) consists of their \wtl{} ($11$ options) and vote in the election ($3! = 6$ options).
Note that because the election mechanism is stochastic, we are forced to confront the issue of how to aggregate stochastic outcomes.
We can construct player $i$'s lottery hierarchically; first sample $a_{-i} \sim x_{-i}$, then independently sample each election outcome given $a_i$ (more details in~\appxref{app:chance}).

We utilize election ranking data from \cite{qian2025maskmirror}, an online variation of the \textit{Lost at Sea} task. The dataset contains 115 elections involving 460 participants, interacting under pseudonymous, animal-based identifiers.
Assuming each voter acted deterministically, we find that approximately 30\% of the 115 elections exhibited pure maximal lottery Nash equilibrium profiles. The remaining 70\% contained at least 1 player with an incentive to deviate, suggesting humans were not voting optimally with respect to their reported preferences.

We select two elections for detailed analysis in~\appxref{appx:lost_at_sea}, one in equilibrium and one not. Election data for one of those elections is shown in Table~\ref{tab:not_pure_ne_group}.
We also demonstrate solving for an equilibrium of the election, assuming Borda as the voting rule.
We then re-use existing NFG solvers to approximate a limiting logit equilibrium (LLE)~\citep{mckelvey1995quantal}, which we then analyze. For example, in one election, \edog{} would prefer itself to win, but \edog{} is low-ranked by everyone. Given that \edog{} is unlikely to win, the LLE strategically suggests \edog{} submit a low \wtl{} to steer the election towards \ebear{}, its third favorite, given its second favorite \ecamel{} is not in the runoff. Note also that the LLEs are mixed strategies\textemdash they mix over many (vote, \texttt{wtl}) combos.

Whereas strategic voting theory assumes well-studied election rules~\citep{myerson1993theory}, our COG formalism enabled a black-box analysis of a bespoke one.

\section{Conclusion}\label{sec:conclusion}

We proposed the first (probabilistic) mixed-strategy Nash equilibrium that generalizes to games with ordinal preferences. The key was to replace expected utility maximization for aggregating payoffs with social choice functions for aggregating preferences. We prove existence, develop practical notions of approximation and algorithms with rates, and demonstrate their use in AI evaluation and (stochastic) election manipulation.

\bibliographystyle{plainnat}
\bibliography{main}

\newpage
\onecolumn
\appendix

\section{Context Ordinal Equilibria}

\subsection{Doubly Probabilistic Social Choice Functions}\label{app:dpscf}

Many voting rules satisfy or can be adapted to satisfy Definition~\ref{def:dpscf} of a doubly pSCF. Scoring rules (e.g. plurality, Borda, veto, etc.) can compute solutions taking a lottery (or distribution) over preference relations as input. So can $k$-approval. Rules are divided into $C1$, $C2$, and $C3$ by~\citet{fishburn1984probabilistic} depending on what information is needed to compute them. $C1$ uses only pairwise majority relationships (e.g., Copeland), $C2$ uses weighted pairwise majority relationships (e.g., ranked pairs, Borda), and then $C3$ is other rules (e.g., Dodgson). The lottery representation enables at least $C1$ and $C2$ rules.

\paragraph{Axioms’ Effect on dpSCFs} Social choice theory is axiomatic. We discuss a few axioms and how they are important to dpSCFs here. For example, if \emph{dictatorship} was possible, player $j$ could induce arbitrarily small mass on some action profile such that it completely overwrites player $i$'s best response regardless of the other outcomes possible; this would be akin to shifting mass onto an outcome with infinite payoff in an NFG. \emph{Clone-invariance} asserts that duplicate actions (i.e., actions that are ranked the same under every co-player action profile) will appear symmetrically in the best response; this ensures equilibria exist that have symmetric mass and payoff across cloned actions. \emph{Paretian} voting rules ensure that if a player ranks an action above another under all co-player action profiles, the aggregate voting rule will as well; translating to COGs, if an action has higher payoff than all other actions under all co-player actions, this action will be a best response. These properties can be taken for granted in classical games that use expected value to aggregate scores, but it is critical and fortuitous that social choice has already developed these tools for use in COGs. Note that Borda fails clone-invariance. It also fails Condorcet consistency, which is a major motivation for works that explore alternatives to the standard RLHF pipeline (see~\citep{maura2025jackpot} which advocates for maximal lotteries).

\subsection{COGs as NFGs}\label{app:cog:nfg}
As mentioned in Section~\ref{sec:veq:existence}, scoring and positional voting rules induce normal-form games. This can be seen by filling out player $i$'s payoff tensor $U_i(\cdot, a_{-i})$ one ``slice'' at a time. For each player $i$ and each possible co-player action profile $a_{-i}$, retrieve player $i$'s preference relation $\rho_i(a_{-i})$ which results in a vector of numerical scores for each of player $i$'s actions: $U_i(\cdot, a_{-i}) \leftarrow \rho_i(a_{-i})$. Scoring and positional voting rules simply take an average of these scores proportional to their representation in the population of votes which is precisely equivalent to computing the expected payoff for each action as is standard in NFG calculations.

\subsection{Correlated Equilibria}\label{app:corr}

Correlated equilibria (CE) are also important in classical game theory, particularly in $n$-player, general-sum settings, and have natural counterparts here.

A correlated equilibrium is a joint distribution $\vx$ over action profiles such that no player has any incentive to unilaterally deviate even after observing their recommended action (sampled from $\vx$). As before, we represent this as a best-response inclusion problem.
Define $\br_i(a_i|\vx) = \nu_i(v_i(x(\mathcal{A}_{-i}|a_i)))$ where the conditional distribution $x(\mathcal{A}_{-i}|a_i) = \vx(a_i,\mathcal{A}_{-i}) / x(a_i)$ and the marginal distribution $x(a_i) = \sum_{a'_{-i} \in \mathcal{A}_{-i}} \vx(a_i,a'_{-i})$. Define scalar-set multiplication as $a \cdot \br_i(a_i|\vx) = \{a \cdot v | v \in \br_i(a_i|\vx) \}$. Then $\mathbf{0} \in x(a_i) \cdot (\br_i(a_i \vert \vx) - e_{a_i})$ for all $i$ represents the condition for a context-ordinal CE where $e_{a_i}$ is the standard Euclidean basis vector.

\begin{figure*}[ht!]
    \centering
    \begin{subfigure}[t]{0.24\textwidth}
        \centering %
        \caption{Classic} %
        \label{fig:nwon_chicken:classic} %
        \includegraphics[width=\linewidth]{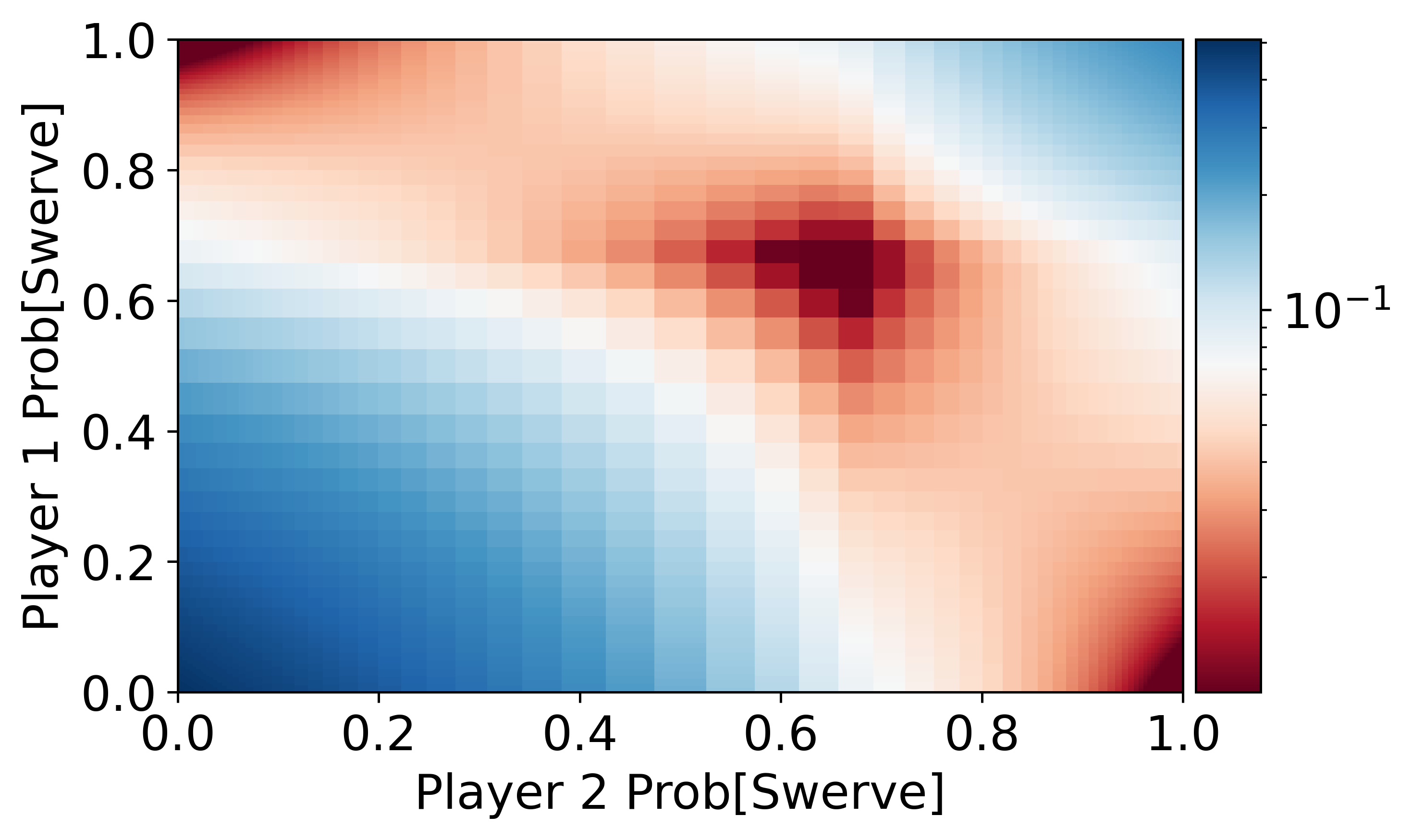}
    \end{subfigure}
    \hfill
    \begin{subfigure}[t]{0.24\textwidth}
        \centering %
        \caption{Classic (EMD)} %
        \label{fig:nwon_chicken:payoff_scoring} %
        \includegraphics[width=\linewidth]{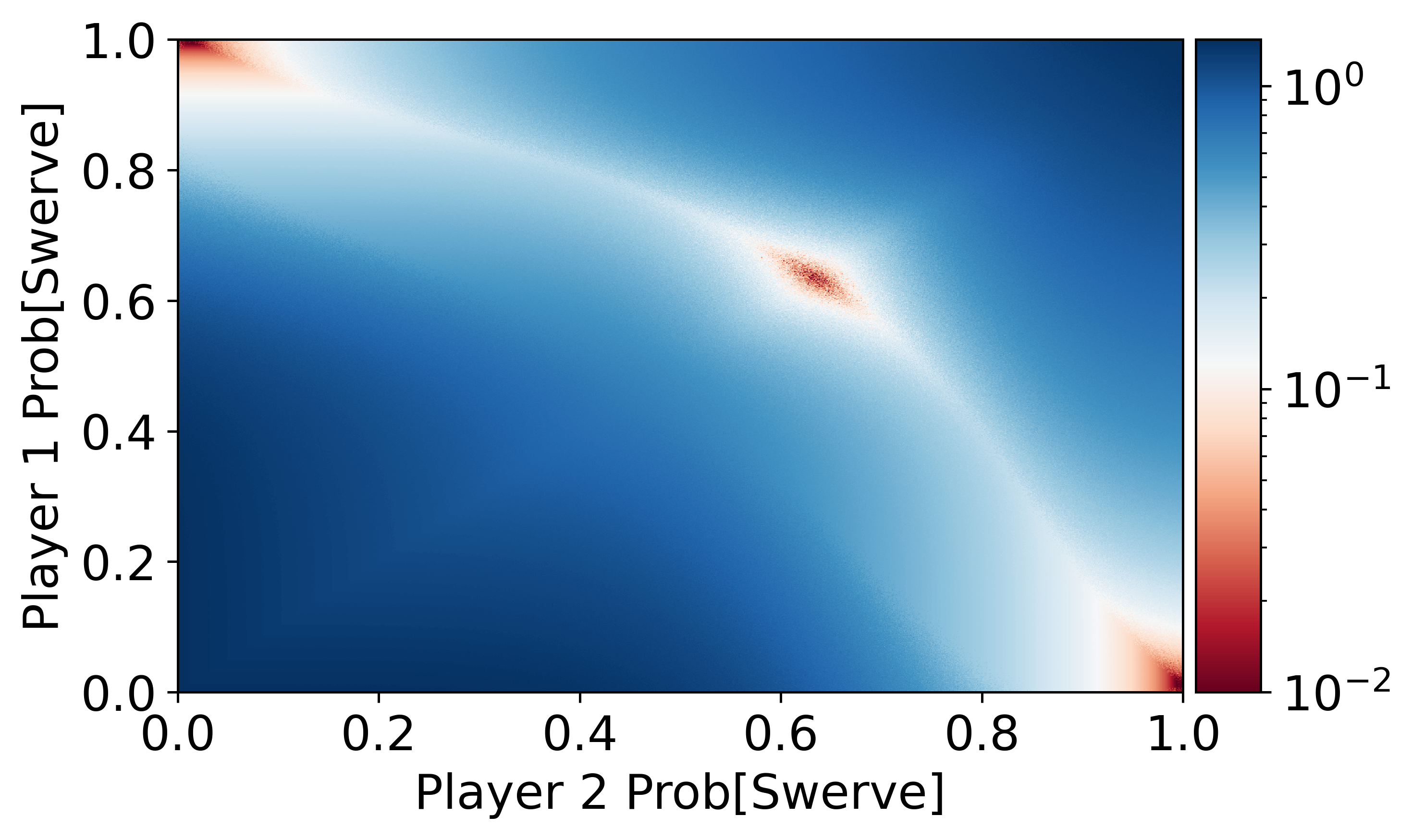}
    \end{subfigure}
    \hfill
    \begin{subfigure}[t]{0.24\textwidth}
        \centering %
        \caption{ML (EMD)} %
        \label{fig:nwon_chicken:ml} %
        \includegraphics[width=\linewidth]{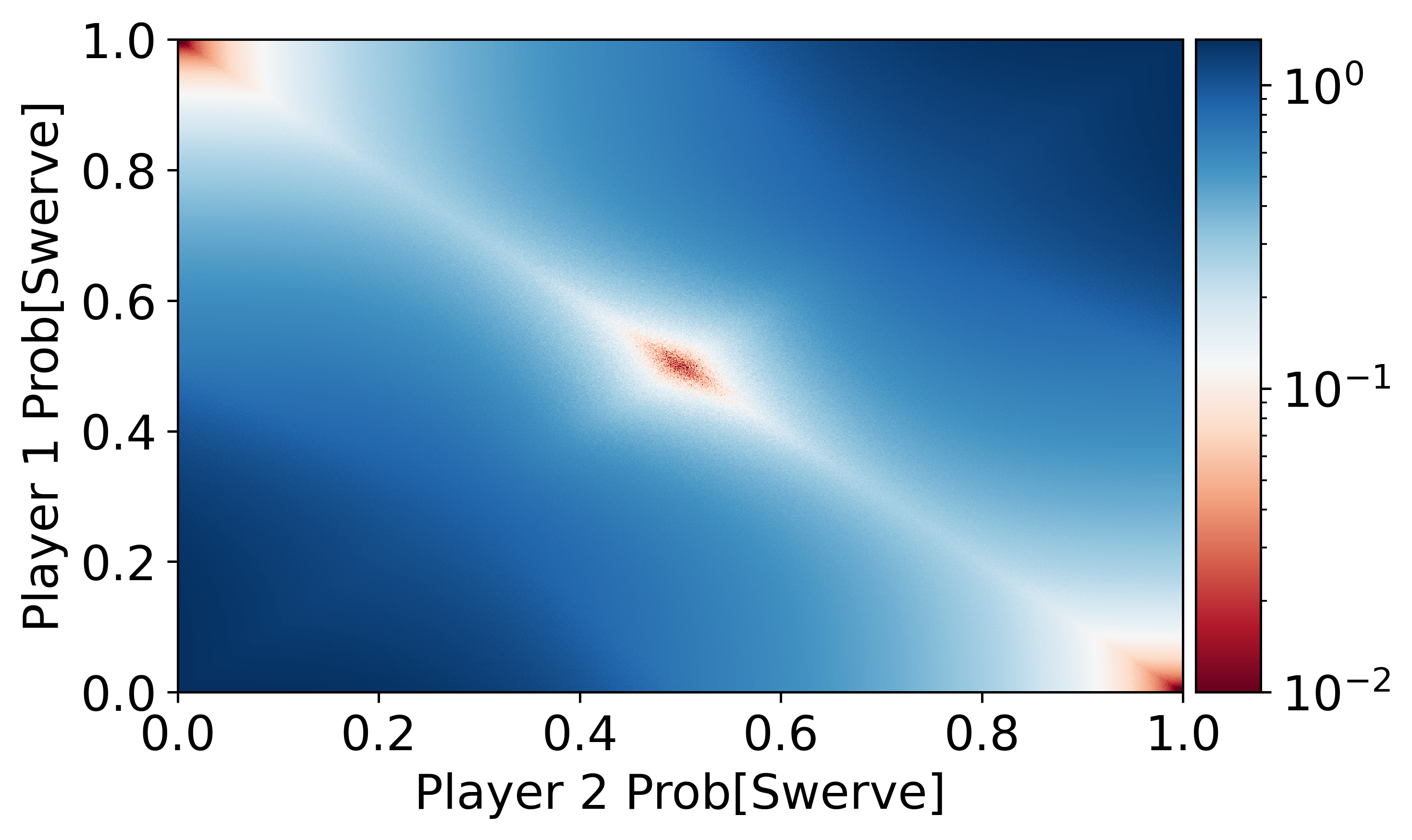}
    \end{subfigure}
    \hfill
    \begin{subfigure}[t]{0.24\textwidth}
        \centering %
        \caption{Borda (EMD)} %
        \label{fig:nwon_chicken:borda} %
        \includegraphics[width=\linewidth]{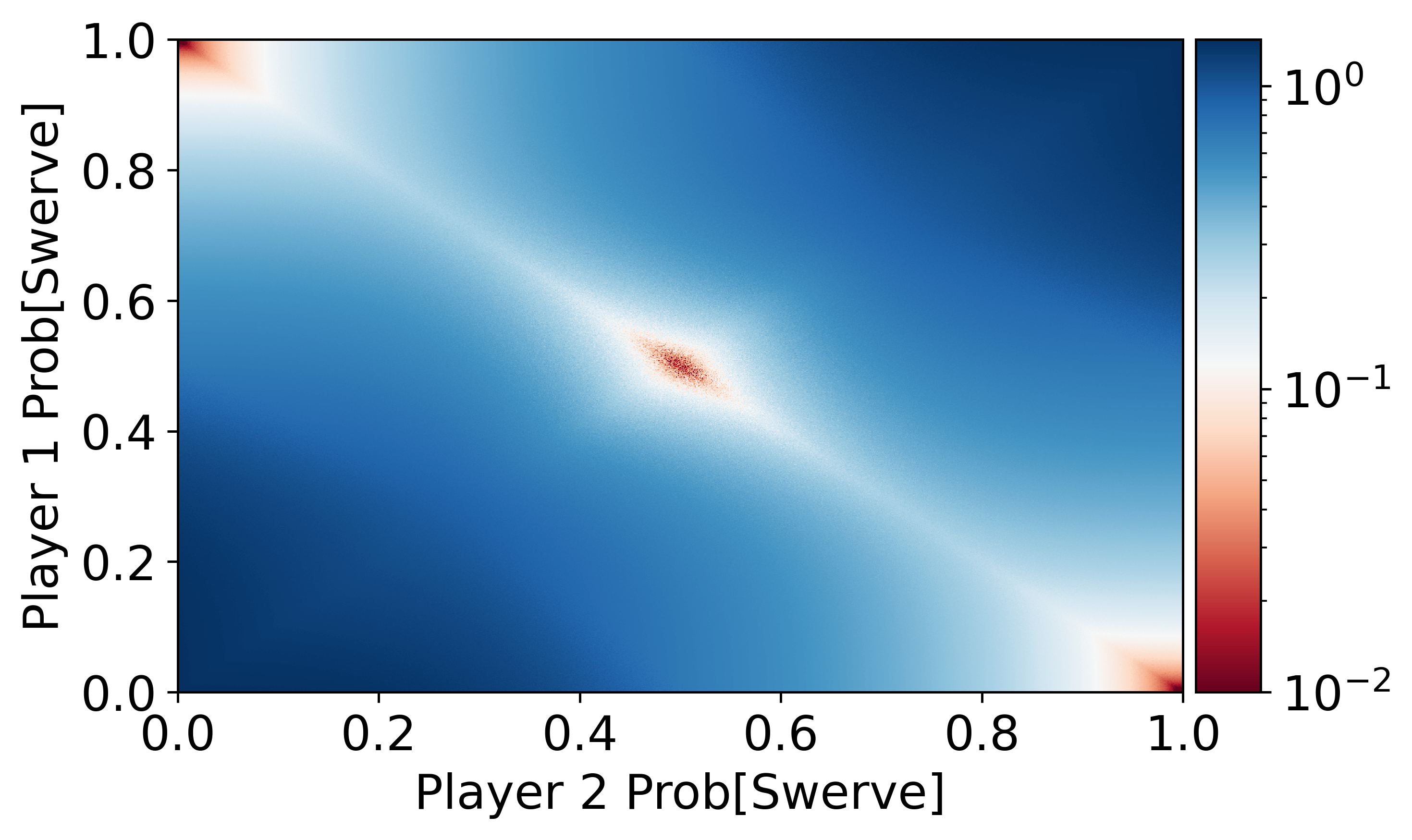}
    \end{subfigure}
    \caption{Exploitability ($\epsilon$) landscapes for a $2$-action (Swerve/Straight) Chicken game using various regularized social choice $\br^{(p=0,q=0.1,\mu=[\sfrac{1}{2}, \sfrac{1}{2}])}$s. Panel~(\subref{fig:nwon_chicken:classic}) displays traditional exploitability based on cardinal payoffs whereas~(\subref{fig:nwon_chicken:payoff_scoring}), (\subref{fig:nwon_chicken:ml}), and (\subref{fig:nwon_chicken:borda}) measure earth mover's distance (EMD). Maximal lotteries (ML) and Borda coincide in two candidate (action) settings.}
    \label{fig:nwon_chicken}
\end{figure*}

Interestingly, \emph{coarse} correlated equilibria (CCE) cannot be (classically) defined in COGs. In both NE and CE, a player considers deviating from their current strategy to an alternative strategy under a fixed context, either a co-player strategy or conditional belief. In a CCE, the expected utility of a joint distribution is compared against a fixed action. This requires aggregating utility \emph{across} contexts with different weights, so a player must effectively be able to compare their own actions under different co-player action profiles. In the COG definition, we specifically assume that a player can only rank their actions under a fixed co-player action profile, which rules out this possibility.

However, we discussed a generalized notion of \emph{external regret} in Section~\ref{sec:approx:winrate} which suggests an alternative route towards CCE. A CCE can instead be defined as a joint distribution such that no player has an incentive to deviate to any fixed strategy (in hindsight), i.e., all players simultaneously experience zero external regret.

\subsection{Exploitability Landscapes}

Figure~\ref{fig:nwon_chicken} shows how CO-NE differ from classical NE on a simple 2-player, 2-action Chicken game. Panels~\subref{fig:nwon_chicken:classic} and~\subref{fig:nwon_chicken:payoff_scoring} display the traditional exploitability landscape and earth mover's distance (EMD) landscape for a classical NE respectively. Panel~\subref{fig:nwon_chicken:ml} shows the EMD landscape for a CO-NE defined using maximal lotteries as the voting rule. Panel~\subref{fig:nwon_chicken:borda} shows the same landscape but for Borda as the voting rule.

\section{Regularized Best Response}\label{app:reg_br}

\begin{algorithm}[ht]
\caption{Regularized Best Response}
\label{alg:reg_br}
\begin{algorithmic}[1]
    \STATE Given: $p$, $\mu_i$, $\nu_i$, $\gamma(\cdot \vert x_{-i}, q)$, $M$
    \STATE $\br_i = \mathbf{0}_{\vert \mathcal{A}_i \vert}$
    \FOR{$m = 1$ to $M$}
        \STATE $\hat{x}_{-i} \sim \gamma(\cdot \vert x_{-i}, q)$
        \STATE $u \sim Cat(\mu_i)$ \alglabel{line:usurper_1}
        \STATE $v_u = \{a_u \succ a_j \,\, \vert \,\, j \in [\vert \mathcal{A}_i \vert], j \ne u \}$ \alglabel{line:usurper_2}
        \STATE $v_u = v_u \cup \{a_j \sim a_k \,\, \vert \,\, j, k \in [\vert \mathcal{A}_i \vert] / \{u\}, j \ne k \}$ \alglabel{line:usurper_3}
        \STATE $V = \{v: 0 \,\, \vert \,\, v \in \mathcal{P}(\mathcal{A}_i)\}$
        \FORALL{$a_{-i} \in \mathcal{A}_{-i}$}
            \STATE $bit ~\sim Bern(p)$
            \STATE $freq = \hat{x}_{-i}(a_{-i})$
            \IF{$bit==1$}
                \STATE $V(v_u) \mathrel{+}= freq$ \alglabel{line:usurper_4}
            \ELSE
                \STATE $V(\rho_i(a_{-i})) \mathrel{+}= freq$ \alglabel{line:no_effect}
            \ENDIF
        \ENDFOR
        \STATE $\br_i \mathrel{+}= \nu_i(V) / M$
    \ENDFOR
    \STATE Output: $\br_i$
\end{algorithmic}
\end{algorithm}

We now describe our approach to regularizing best responses. Algorithm~\ref{alg:reg_br} provides pseudocode. Recall from Definition~\ref{def:br} that every action profile $a_{-i}$ induces a vote $\rho_i(a_{-i})$ or \emph{ballot type} and this vote is represented in a population of votes with frequency $x_{-i}(a_{-i})$. We will represent random perturbations as randomly swapping votes for votes from another distribution dependent on the given $\target_i$. Consider the following random process. Draw a random \emph{usurper} from $\target_i$. Set the \emph{usurper ballot} such that the usurper is ranked strictly first; the remaining players can be ranked arbitrarily but we set all others tied for second in experiments. For every co-player action profile, replace the original ballot with the usurper ballot with probability $p$. Apply the (dpSCF) voting rule to this perturbed population of votes to compute a sampled best response. Take the expectation over these sampled best responses to return an expected best response.

If $p = 0$, then the best response remains the same as before. If $p = 1$, all ballots are replaced by the usurper ballot and any voting rule satisfying majority rule will select the usurper as the sampled best response. As the usurper is selected according to $\target_i$, the expected best response will also be equal to $\target_i$.

At this point, we have introduced two new parameters, $p$ and $\mu_i$. Denote this parameterized best response by $\br_i^{(p,\mu_i)}$.

To achieve a single-valued best response, we will replace all non-singleton best responses, which are necessarily convex subsets of the simplex, with their centroid, i.e, the uniform distribution over the winning candidates. Note this now means $\br^{(p,\mu_i)}_i$ is no longer u.h.c. as it could jump at ties, e.g., $A \succ B \rightarrow B \sim A \rightarrow B \succ A$ leads to $[1, 0] \rightarrow [0.5, 0.5] \rightarrow [0, 1]$.
Recall one of our desiderata is to ensure continuity, a stronger condition than u.h.c.

To render the best response $\br^{(p,\mu_i)}_i(x_{-i})$ continuous with respect to $x_{-i}$, we can introduce a \emph{trembling hand} by taking the expectation over a Dirichlet perturbation, $Dir(\alpha)$, with density $\rho$ where $\alpha = \mathbf{1} + \sfrac{1}{q} \, x_{-i}$ with $q > 0$:
\begin{align}
    \br_i^{(p,q,\mu_i)}(x_{-i}) &= \mathbb{E}_{x'_{-i}\sim Dir(\alpha)}[\br_i^{(p,\mu_i)}(x'_{-i})] = \int_{\mathcal{X}_{-i}} \br_i^{(p,\mu_i)}(x'_{-i}) \rho(x'_{-i} \vert \alpha) dx'_{-i}.
\end{align}
For any fixed $q > 0$, the Dirichlet density $\rho(x'_{-i} \vert \alpha)$ is continuous in its parameters $\alpha$ and thus also in $x_{-i}$. Notice that $\br_i^{(p,\mu_i)}(x'_{-i})$ is independent of $x_{-i}$, therefore the integrand is continuous in $x_{-i}$ for every fixed $x'_{-i}$. Because $q > 0$, this implies $\alpha \ge 1$ and so the density $\rho$ is uniformly bounded.

Combined with the fact that best responses are probability distributions (and thus bounded), the integrand is dominated by a constant integrable function.
We then invoke the Dominated Convergence Theorem~\citep{folland2013real}[p. 56, Theorem 2.27] to conclude that $\br^{(p,q,\mu_i)}_i(x_{-i})$ is continuous.

Lastly, note that as $\lim_{q \rightarrow 0^{+}}$, $Dir(\alpha)$ converges weakly to a Dirac delta at $x_{-i}$, meaning we achieve continuity while recovering an element of the original best response set in the limit. The Dirichlet also as full support over the simplex for $q > 0$, consistent with the ``trembling hand'' interpretation of sampling all distributions with positive probability.
In summary, $p$ regularizes $\br_i$ towards $\mu_i$; $q$ smooths.

\begin{definition}[Regularized Best Response]
Let $\br_i^{(p,q,\target_i)}(x_{-i})$ denote the expectation of a dpSCF (Def.~\ref{def:dpscf}), in which, with probability $p$, ballot types are replaced with ballots top-ranked by an action sampled from $\target_i$, non-singleton dpSCF outputs are replaced with their centroid, and the underlying voting distribution $x_{-i}$ is replaced with a full-support distribution $\gamma$ such that $\gamma$ converges to a Dirac delta distribution on $x_{-i}$ as $q \rightarrow 0$ and uniform as $q \rightarrow \infty$.
\end{definition}

\begin{proposition}
The regularized best response function $\br_i^{(p,q,\target_i)}(x_{-i})$ with $p \in [0, 1]$, $q > 0$, and $\target_i \in \Delta^{\mathcal{A}_i}$ satisfies desiderata~\peqref{des:target},~\peqref{des:unique},~\peqref{des:cont}. In the limit $q \rightarrow 0$,~\peqref{des:no_effect} is satisfied.
\end{proposition}

\begin{figure}[t!]
    \centering
    \includegraphics[width=0.4\linewidth]{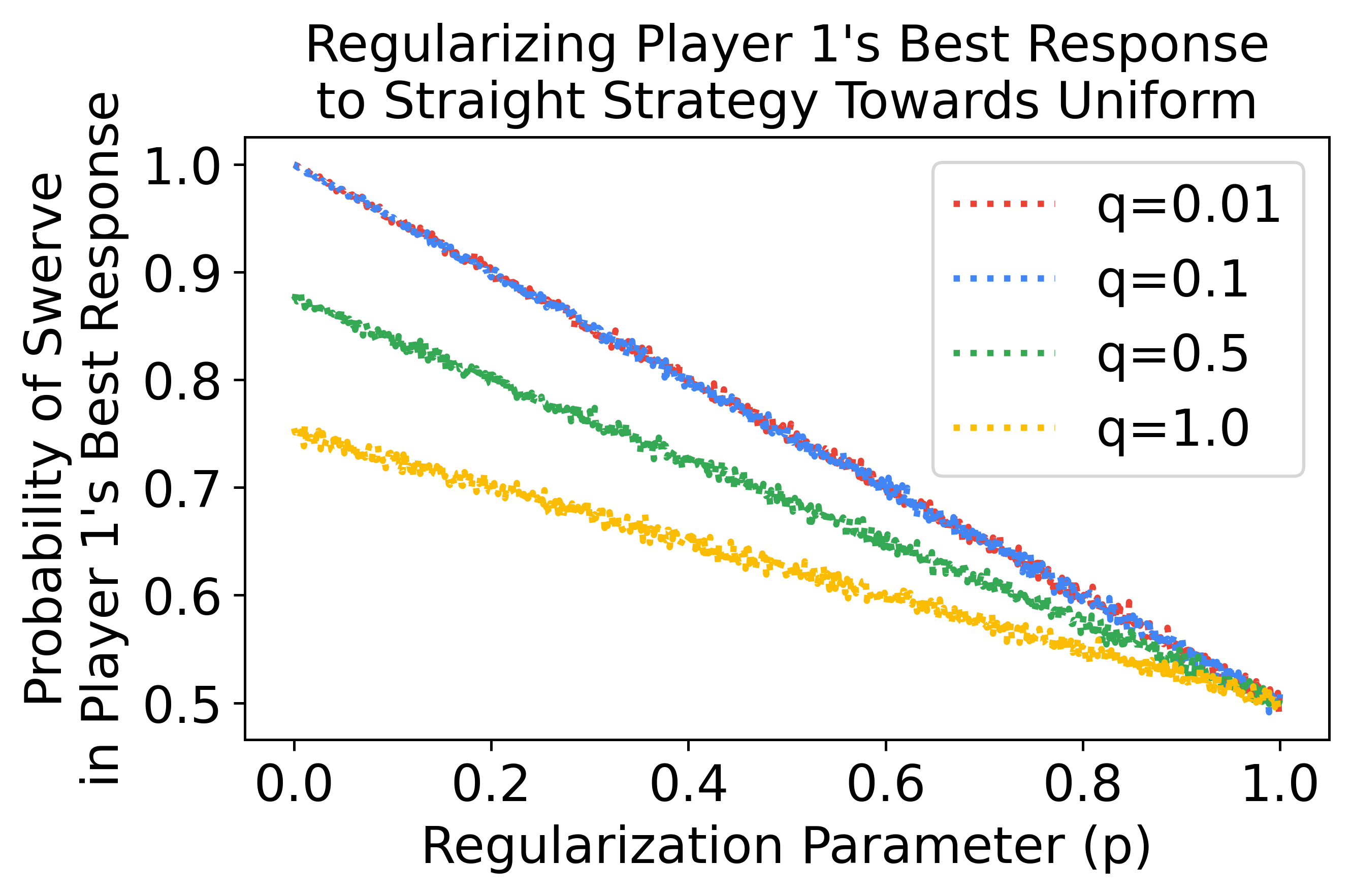}
    \caption{Best Response Regularization in Chicken.}
    \label{fig:br_ref}
\end{figure}

\subsection{Empirical Demonstration}\label{app:br_reg}

Figure~\ref{fig:br_ref} displays the results of a simple experiment, regularizing a maximal lottery best response in a chicken game towards a uniform strategy. A chicken game is a symmetric game in which you prefer to go ``straight'' unless your co-player does, in which case you would rather ``swerve'' to avoid a collision, i.e., $\rho$(``swerve'') $\rightarrow$ ``straight'' $\succ$ ``swerve'' and $\rho$(``straight'') $\rightarrow$ ``swerve'' $\succ$ ``straight''.
We find that setting $q=0.1$ is sufficient to achieve both continuity and correctness of the best response over the range of $p$.

For Figures~\ref{fig:nwon_chicken:classic} and~\ref{fig:nwon_chicken:payoff_scoring}, the precise payoff matrices used for player $1$ and $2$ were $U^{(1)} = \begin{bmatrix}
3/4 & 1/2 \\
1 & 0
\end{bmatrix}$
and 
$U^{(2)} = \begin{bmatrix}
3/4 & 1 \\
1/2 & 0
\end{bmatrix}$
respectively.

\section{Distortion}\label{app:distortion}

Traditionally, distortion expresses the ratio of the maximum social welfare possible given knowledge of voters' cardinal scores for each candidate to the social welfare achieved by a voting rule applied to those same voters' preferences. Bounds on distortion are typically derived assuming voters' cardinal scores are non-negative and sum-to-1. Many voting rules process votes as ordinal rankings which destroys the more detailed cardinal information which means distortion is commonly strictly greater than 1.

We derive a lemma below that extends a distortion bound given each voter $k$'s scores sum-to-$s_k$.

In what follows, election outcomes may be probabilistic, which we represent as a distribution $\vx$ over candidates. The expected welfare $SW^s$ of this election outcome given access to the scaled voting scores is
\begin{align}
    SW^s(\vx) &= \frac{1}{K}\sum_k s_k v_k^\top \vx
\end{align}
where $K$ is the number of voters and $v_k$ is a vector of voter $k$'s unit-normalized scores for each candidate.

\begin{replemma}{lemma:scaled_distortion}
Assume $\nu$ is a voting rule that only processes ordinal rankings. Then
\begin{align}
d(\nu, \vs) \le \kappa d(\nu)    
\end{align}
where $\kappa = \frac{\max_k s_k}{\min_k s_k}$ and $d(\nu)$ is an upper bound on the distortion of the voting rule $\nu$ when each voter's ballot scores are non-negative and sum-to-1.
\end{replemma}
\begin{proof}
Let $\vx_{cs}$ be the distribution over candidates that maximizes welfare assuming access to the voters' \emph{scaled} cardinal scores:
\begin{align}
    \vx_{cs} &= \argmax_{\vx_{cs}} \frac{1}{K} \sum_k s_k v_k^\top \vx_{cs}.
\end{align}
Let $\vx_c$ be the distribution over candidates that maximizes welfare assuming access to the voters' \emph{un-scaled}, i.e., unit-normalized, cardinal scores:
\begin{align}
    \vx_{c} &= \argmax_{\vx_{c}} \frac{1}{K} \sum_k v_k^\top \vx_{c}.
\end{align}
The social welfare of $\vx_{cs}$ is then upper bounded as
\begin{align}
    SW^s(\vx_{cs}) &= \frac{1}{K} \sum_k s_k v_k^\top \vx_{cs} \le \frac{1}{K} (\max_k s_k) (\sum_k v_k^\top \vx_{cs}) \le \frac{1}{K} (\max_k s_k) (\sum_k v_k^\top \vx_{c}).
\end{align}

Let $x_{\nu}$ be the distribution over candidates returned by the voting rule $\nu$. The social welfare of $\vx_{\nu}$ is then lower bounded as
\begin{align}
    SW^s(\vx_{\nu}) &= \frac{1}{K} \sum_k s_k v_k^\top \vx_{\nu} \ge \frac{1}{K} (\min_k s_k) (\sum_k v_k^\top \vx_{\nu}).
\end{align}

The distortion is the ratio of the two and is upper bounded as
\begin{align}
d(\nu, \vs) &= \frac{SW^s(\vx_{cs})}{SW^s(\vx_{\nu})} \le \frac{\max_k s_k}{\min_k s_k} \frac{\sum_k v_k^\top \vx_{cs}}{\sum_k v_k^\top \vx_{\nu}}
\\ &\le \frac{\max_k s_k}{\min_k s_k} \frac{\sum_k v_k^\top \vx_{c}}{\sum_k v_k^\top \vx_{\nu}}
\\ &= \frac{\max_k s_k}{\min_k s_k} d(\nu).
\end{align}

\end{proof}

Given a strategy profile $\vx$ for an $n$-player COG, we can use Lemma~\ref{lemma:scaled_distortion} above to understand the suboptimality of a player's best response if they opt for a social choice rule (dpSCF) rather than a traditional best response using expected utility theory. To apply Lemma~\ref{lemma:scaled_distortion}, we will assume the underlying payoffs of the game are non-negative with strictly positive sum. Note that an affine transformation of the payoff matrix does not change the set of equilibria so this assumption is without loss of generality.

Simply let $\vs_i$ be a vector containing the payoff sums for player $i$ under each possible joint action of the co-players. Then apply Lemma~\ref{lemma:scaled_distortion} by looking up the distortion bound for the chosen voting rule $\nu_i$.

\begin{reptheorem}{lemma:scaled_distortion}
Assume $\nu$ is a voting rule that only processes ordinal rankings. Then distortion
\begin{align}
d^+(\nu, \vs) \le \kappa^+ + (\min_k s_k) d^+(\nu)
\end{align}
where $\kappa = (\max_k s_k - \min_k s_k)$ and $d^+(\nu)$ is an upper bound on the \emph{additive} distortion of the voting rule $\nu$ when each voter's ballot scores are non-negative and sum-to-1.
\end{reptheorem}
\begin{proof}
As in Lemma~\ref{lemma:scaled_distortion}, let $SW^s$ denote the social welfare function with scaled votes. Let $SW$ denote the social welfare function assuming unit-normalized scores, which implies $SW \le 1$. We are interested in the additive distortion:
\begin{align}
    d^+(\nu, \vs) &= SW^s(\vx_{cs}) - SW^s(\vx_{\nu})
    \\ &\le (\max_k s_k) SW(\vx_{c}) - (\min_k s_k)  SW(\vx_{\nu})
    \\ &= (\max_k s_k - \min_k s_k) SW(\vx_{c}) + (\min_k s_k) ( SW(\vx_{c}) - SW(\vx_{\nu}))
    \\ &\le \kappa^+ + (\min_k s_k) d^+(\nu).
\end{align}
\end{proof}

\begin{theorem}\label{thm:reg_distortion}
The additional distortion introduced by our regularization process is 
\begin{align}
    d^+(\nu_{p,i}, u_i, x_{-i}) &\le \kappa^+ + (\min_k s_k) (d^+(\nu) + \big( 1 - (1-p)^{\vert \mathcal{A}_{-i} \vert} \big)).
\end{align}
By Bernoulli's inequality, the quantity $\big( 1 - (1-p)^{\vert \mathcal{A}_{-i} \vert} \big)$ behaves as $\vert \mathcal{A}_{-i} \vert p$ for small $p$.
\end{theorem}
\begin{proof}
The regularized voting process replaces each vote with a usurper vote with probability $p$. Therefore, with probability $(1-p)^{\vert \mathcal{A}_{-i} \vert}$, no votes are replaced. The voting rule returns $\vx_{\nu}$ in this case. Taking the expectation of these outputs, we can then determine
\begin{align}
    \vx_{\nu_p} &= (1-p)^{\vert \mathcal{A}_{-i} \vert} \vx_{\nu} + \big( 1 - (1-p)^{\vert \mathcal{A}_{-i} \vert} \big) \vz
\end{align}
where $x_{\nu_p}$ denotes the perturbed output (regularized best response) and $\vz$ denotes the expected output under the remaining perturbation events. The distortion of our perturbed voting rule, denoted $\nu_p$, can be decomposed into the distortion of the original voting rule $\nu$ and the gap between the perturbed and original:
\begin{align}
    d^+(\nu_p) &= \underbrace{SW(\vx_{cs}) - SW(\vx_{\nu})}_{\text{distortion}} + \underbrace{SW(\vx_{\nu}) - SW(\vx_{\nu_p})}_{\text{perturbation error}}
    \\ &\le d^+(\nu) + SW(\vx_{\nu}) - SW(\vx_{\nu_p})
    \\ &= d^+(\nu) + SW(\vx_{\nu}) - SW((1-p)^{\vert \mathcal{A}_{-i} \vert} \vx_{\nu} + \big( 1 - (1-p)^{\vert \mathcal{A}_{-i} \vert} \big) \vz)
    \\ &= d^+(\nu) + \big( 1 - (1-p)^{\vert \mathcal{A}_{-i} \vert} \big) \big( SW(\vx_{\nu}) - SW(\vz) \big) \label{eqn:sw_lin}
    \\ &\le d^+(\nu) + \big( 1 - (1-p)^{\vert \mathcal{A}_{-i} \vert} \big) \label{eqn:sw_one}
    \\ &\le d^+(\nu) + \vert \mathcal{A}_{-i} \vert p
\end{align}
where~\peqref{eqn:sw_lin} follows from linearity of social welfare,~\peqref{eqn:sw_one} from social welfare being bounded to $[0,1]$, and the last step from Bernoulli's inequality.

Plugging this result into Theorem~\ref{lemma:scaled_distortion} achieves the claim.
\end{proof}

\subsection{Weakened Fictitious-Play (WFP)}

We consider fictitious-play run for $T$ rounds with approximate best responses whose error is upper bounded by $d^+(\nu, \vs)$. We can trace the argument put forth in~\citep{conitzer2009approximation} to prove the following. 

\begin{repcorollary}{cor:cofp}[Theorem 1~\citep{conitzer2009approximation}]
The fictitious-play profile $\vx_T = [x^{(1)}_T, x^{(2)}_T]$ after $T$ rounds is an $\epsilon_T$-NE where $\epsilon_T \le \frac{T+1}{2T} + \frac{1}{2} d^+(\nu, \vs)$.
\end{repcorollary}
\begin{proof}
We simply trace~\citep[Theorem 1]{conitzer2009approximation} with approximate best responses.

By symmetry, it suffices to show that $x^{(1)}_T = \frac{1}{T} \sum_{t=1}^T x^{(1)}_t$ is an $\epsilon_T$-best response to $x^{(2)}_T = \frac{1}{T} \sum_{t=1}^T x^{(2)}_t$. Let $\br_1$ be a best response to $x^{(2)}_T$. The corresponding best-response utility for player $1$ is $u_1(\br_1, x^{(2)}_T) = \frac{1}{T} \sum_{t=1}^T u_1(\br_1, x^{(2)}_t)$. For $2 \le t' \le T + 1$, because $x^{(1)}_{t'}$ is a best response to $x^{(2)}_{t' - 1}$, all utilities are non-negative, and the utilities are bilinear, we have
\begin{align}
u_1(x^{(1)}_{t'}, x^{(2)}_{T}) &= u_1(x^{(1)}_{t'}, \frac{1}{T} \sum_{t=1}^T x^{(2)}_t)
\\ &= \sum_{t=1}^T (\sfrac{1}{T}) u_1(x^{(1)}_{t'}, x^{(2)}_t)
\\ &= \sum_{t=1}^{t' - 1} (\sfrac{1}{T}) u_1(x^{(1)}_{t'}, x^{(2)}_t) + \sum_{t=t'}^{T} (\sfrac{1}{T}) u_1(x^{(1)}_{t'}, x^{(2)}_t)
\\ &\ge \sum_{t=1}^{t' - 1} (\sfrac{1}{T}) u_1(x^{(1)}_{t'}, x^{(2)}_t)
\\ &= \big( \frac{t' - 1}{T} \big) u_1(x^{(1)}_{t'}, \frac{1}{(t' - 1)} \sum_{t=1}^{t' - 1} x^{(2)}_t)
\\ &= \big( \frac{t' - 1}{T} \big) u_1(x^{(1)}_{t'}, x^{(2)}_{t' - 1})
\\ &\ge \big( \frac{t' - 1}{T} \big) [u_1(\br_1, x^{(2)}_{t' - 1}) - d^+(\nu, \vs)] \label{eqn:better_response}
\\ &= u_1(\br_1, \frac{1}{T} \sum_{t=1}^{t' - 1} x^{(2)}_t) - \sum_{t=1}^{t' - 1} (\sfrac{1}{T}) d^+(\nu, \vs)
\\ &= \sum_{t=1}^{t' - 1} (\sfrac{1}{T}) [u_1(\br_1, x^{(2)}_t) - d^+(\nu, \vs)]
\end{align}
where~\peqref{eqn:better_response} follows by the fact that $x^{(1)}_{t'}$ is an approximate best response to $x^{(2)}_{t' - 1}$; most importantly, it at least $-d^+(\nu, \vs)$ better than any other strategy (including $\br_1$) in responding to $x^{(2)}_{t' - 1}$.

Continuing, for the case where player 1 plays $x^{(1)}_T$, we have
\begin{align}
    u_1(x^{(1)}_T, x^{(2)}_T) &=  u_1(\frac{1}{T} \sum_{t'=1}^T x^{(1)}_{t'}, x^{(2)}_T)
    \\ &= \sum_{t'=1}^T (\sfrac{1}{T}) u_1(x^{(1)}_{t'}, x^{(2)}_T)
    \\ &= (\sfrac{1}{T}) u_1(x^{(1)}_{t'=1}, x^{(2)}_T) + \sum_{t'=2}^T (\sfrac{1}{T}) u_1(x^{(1)}_{t'}, x^{(2)}_T)
    \\ &\ge 0 + \sum_{t'=2}^T (\sfrac{1}{T}) u_1(x^{(1)}_{t'}, x^{(2)}_T)
    \\ &\ge \sum_{t'=1}^T (\sfrac{1}{T}) \sum_{t=1}^{t' - 1} (\sfrac{1}{T}) [u_1(\br_1, x^{(2)}_t) - d^+(\nu, \vs)] \quad \text{summand equals $0$ for $t'=1$}
    \\ &= (\sfrac{1}{T^2}) \sum_{t=1}^{T-1} \sum_{t'=t+1}^{T} [u_1(\br_1, x^{(2)}_t) - d^+(\nu, \vs)] \quad \text{re-index sum of lower triangular ($t' \times t$) matrix}
    \\ &= (\sfrac{1}{T^2}) \sum_{t=1}^{T} (T - t) [u_1(\br_1, x^{(2)}_t) - d^+(\nu, \vs)].
\end{align}

On the other hand,
\begin{align}
u_1(\br_1, x^{(2)}_T) &= u_1(\br_1, \frac{1}{T} \sum_{t=1}^T x^{(2)}_t)
\\ &= (\sfrac{1}{T}) \sum_{t=1}^T u_1(\br_1, x^{(2)}_t)
\\ &= (\sfrac{1}{T^2}) \sum_{t=1}^T T u_1(\br_1, x^{(2)}_t).
\end{align}

It follows that the suboptimality for player $1$ of playing $x^{(1)}_T$ is
\begin{align}
u_1(\br_1, x^{(2)}_T) - u_1(x^{(1)}_T, x^{(2)}_T) &\le (\sfrac{1}{T^2}) \sum_{t=1}^{T} t u_1(\br_1, x^{(2)}_t) + (\sfrac{1}{T^2}) \sum_{t=1}^{T} (T - t) d^+(\nu, \vs)
\\ &\le (\sfrac{1}{T^2}) \sum_{t=1}^{T} t + (\sfrac{1}{T^2})  d^+(\nu, \vs) \sum_{t=1}^{T} (T - t)
\\ &= (\sfrac{1}{T^2}) (T + 1) (\sfrac{T}{2}) + (\sfrac{1}{T^2})  d^+(\nu, \vs) \sum_{t=0}^{T-1} t
\\ &= (\sfrac{1}{T^2}) (T + 1) (\sfrac{T}{2}) + (\sfrac{1}{T^2})  d^+(\nu, \vs) (T - 1) (\sfrac{T}{2})
\\ &= \frac{(T + 1)}{(2T)} + \frac{(T - 1)}{(2T)} d^+(\nu, \vs).
\end{align}
\end{proof}

\subsubsection{Distortion of Regularized Best Responses}

The bounds above are derived for unregularized best responses. There exists a new result in the distortion literature that examines multiplicative distortion under a different perturbation model~\citep{alipour2026utilitarian}. If this line of work were to continue and uncover additive bounds, we might be able to apply them here to obtain bounds for a different form of regularized best responses.

\subsubsection{First Order Stochastic Dominance}

We clarify a relationship between first order stochastic dominance (FSD) and social choice based best responses. Within the context of COGs, let $x_i$ and $x_i'$ be two mixed strategies from which we can sample actions $a_i$ and $a_i'$. Let $u_i: \mathcal{A} \rightarrow \mathbb{R}$ be any isotone (order-preserving) utility function that is consistent with the given preference relation $\rho_i$. In other words, if $a_i \succ a_i'$ in the context of $a_{-i}$, then $u_i(a_i, a_{-i}) > u_i(a'_i, a_{-i})$ and if $a_i \sim a_i'$ in the context of $a_{-i}$, then $u_i(a_i, a_{-i}) = u_i(a'_i, a_{-i})$. Then $x$ strictly FSD-dominates $x_i'$ if $\mathbb{E}[u_i(x_i, x_{-i})] > \mathbb{E}[u_i(x_i', x_{-i})]$ holds for every $u_i$ where $u_i$ has been extended to act on distributions in the standard way.

Consider $x_i'$ player $i$'s current strategy. Assume $x_i$ strictly FSD-dominates $x_i'$ as above. Under which voting rules does player $i$ have a strict incentive to deviate?

Within social choice theory, the property that an aggregation rule never selects a lottery that is stochastically dominated by another is known as SD-efficiency.

Both positional scoring rules and Maximal Lotteries are SD-efficient, hence, no strategies in their induced best responses sets will ever be FSD-dominated by another strategy.

\section{Complexity}\label{app:complexity}

It is natural to attempt to understand the complexity of computing voting equilibria. Given two-player, zero-sum games represent a natural complexity boundary in the classical payoff setting, we first explore whether this boundary translates to the voting setting.

\subsection{Two-Player, Zero-Sum}

First off, it is not immediately possible to translate the notion of 2-player, zero-sum to a COG. First, there no longer exist payoffs that can be summed. Second, players are not required to express preferences over unilateral changes in actions by the other player.

As mentioned in Section~\ref{sec:veq:existence}, score and positional voting rules induce normal-form games with payoffs from which we can then analyze traditional \tpzs{} definitions. However, we show that COGs defined to be intuitively adversarial fail these necessary conditions. In particular, adversarial COGs are neither \emph{harmonic} nor \emph{strictly-competitive}. Harmonic games~\citep{candogan2011flows,legacci2024no} generalize the classical \tpzs{} definition in a way that relies on only the weighted \emph{response graph}, a graph of joint action nodes with arrows indicating favorable (including ties) deviations for players. Strictly competitive games~\citep{adler2009note} are \tpzs{} up to shift and scale of each player's payoffs.

We provide a practical example of this phenomenon in Figure~\ref{fig:adv_game} motivated by game-theoretic evaluation~\citep{liure,marris2025deviation}, specifically \emph{Nash averaging}~\citep{balduzzi2018re}. In that setting, one often considers games where it is unclear how to compare performance on one task X (e.g., measured with perplexity) with another task Y (e.g., measured with accuracy), whereas ranking models on a common task is clear.

\begin{figure*}[ht!]
    \centering
    \caption{Adversarial COGs are not necessarily strictly competitive, nor harmonic, however, they are preference-zero-sum. Preferences over agents for each task are indicated in~(\subref{fig:table_agent}) by \textcolor{goldcolor}{gold}, \textcolor{silvercolor}{silver}, and \textcolor{bronzecolor}{bronze}. Preferences over tasks are directly opposed and are indicated in~(\subref{fig:table_task}) by $0$ (not preferred) and $1$ (preferred). The gold arrow in~(\subref{fig:table_agent}) highlights that the agent player can improve their outcome by switching from agent C (silver) to agent A (gold) when playing task X. Assuming Borda as the voting rule (and the payoffs induced in its associated NFG representation), the corresponding black arrow in~(\subref{fig:table_task}) shows the task player is indifferent (although recall preferences over agents are not explicitly represented in a COG for the task player). Hence, this game is not strictly competitive either. The associated (weighted) response graph is shown in~(\subref{fig:response_graph}); weights are calculated again assuming Borda. \appxref{app:zero_sum:harmonic} shows that there does not exist a node weighting that makes the net flow at every node zero, ruling out the game as \emph{harmonic}. However, if one reverses all horizontal edges in the graph (gray), it is possible to construct a common payoff game (values in gray) with the same (\emph{unweighted}) response graph, proving the game is \emph{preference}-zero-sum.} %
    \label{fig:adv_game}

    \begin{subfigure}[t]{0.32\textwidth} %
        \centering %
        \caption{Agent Preferences} %
        \label{fig:table_agent} %
        \vspace{0.5cm}
        \begin{tabular}{|c||c|c|} %
        \cline{2-3} %
        \multicolumn{1}{c|}{} & \multicolumn{2}{c|}{\textbf{Tasks}} \\ %
        \hline
        \textbf{Agents} & \textbf{X} & \textbf{Y} \\ %
        \hline
        \textbf{A} & \gold \tikzmark{ax_a} & \silver \\
        \hline
        \textbf{B} & \bronze & \gold \\
        \hline
        \textbf{C} & \silver \tikzmark{cx_a} & \bronze \\
        \hline
        \end{tabular}
        \begin{tikzpicture}[overlay, remember picture]
            \coordinate (startPoint) at ($(cx_a.west) - (0.1,0.0cm)$);
            \coordinate (endPoint)   at ($(ax_a.west) - (0.1,0.0cm)$);
            \draw[
                ->,
                goldcolor,
                very thick,
                line cap=round,
                out=135,
                in=225,
                looseness=1.5
            ] (startPoint) to (endPoint);
    
        \end{tikzpicture}
    \end{subfigure}%
    \hfill %
    \begin{subfigure}[t]{0.32\textwidth} %
        \centering %
        \caption{Task Preferences} %
        \label{fig:table_task} %
        \vspace{0.5cm}
        \begin{tabular}{|c||c|c|} %
        \cline{2-3} %
        \multicolumn{1}{c|}{} & \multicolumn{2}{c|}{\textbf{Tasks}} \\ %
        \hline
        \textbf{Agents} & \textbf{X} & \textbf{Y} \\ %
        \hline
        \textbf{A} & 0 \tikzmark{ax_t} & 1 \\
        \hline
        \textbf{B} & 1 & 0 \\
        \hline
        \textbf{C} & 0 \tikzmark{cx_t} & 1 \\
        \hline
        \end{tabular}
        \begin{tikzpicture}[overlay, remember picture]
            \coordinate (startPoint) at ($(cx_t.west) - (0.15,0.0cm)$);
            \coordinate (endPoint)   at ($(ax_t.west) - (0.15,0.0cm)$);
            \draw[
                ->,
                black,
                very thick,
                line cap=round,
                out=135,
                in=225,
                looseness=1.5
            ] (startPoint) to (endPoint);
    
        \end{tikzpicture}
    \end{subfigure}
    \begin{subfigure}[t]{0.32\textwidth} %
        \centering %
        \caption{(Weighted) Response Graph} %
        \label{fig:response_graph} %
        \vspace{-0.5cm}
        \begin{tikzpicture}[
            scale=0.8, transform shape,
            node distance=5mm and 10mm,
            mynode/.style={
                circle,       %
                draw,         %
                thick,        %
                minimum size=1mm, %
                inner sep=1pt,
                font=\bfseries %
            },
            myedge/.style={
                ->,           %
                thick,        %
                black,        %
                shorten >=2pt, %
                shorten <=2pt  %
            },
            myreverseedge/.style={
                <-,           %
                dashed,
                thick,
                gray,
                shorten >=2pt,
                shorten <=2pt
            },
            mylabel/.style={
                font=\small,  %
                midway,       %
                fill=none,   %
                inner sep=1pt %
            },
            extlabel/.style={
                font=\tiny, %
                gray,
                anchor=south west, %
                inner sep=0.5pt,   %
                outer sep=0pt      %
            }
            ]
            \node[mynode] (AX) {AX};
            \node[mynode, right=of AX] (AY) {AY}; %
            \node[mynode, below=of AX] (BX) {BX};
            \node[mynode, right=of BX] (BY) {BY}; %
            \node[mynode, below=of BX] (CX) {CX}; %
            \node[mynode, right=of CX] (CY) {CY};

            \draw[myedge] (AX) to node[mylabel, above]{+1} (AY); %
            \draw[myreverseedge, bend right=15] (AX) to node[mylabel, below]{} (AY);
            \draw[myedge] (BX) to node[mylabel, right]{+2} (AX);
            \draw[myreverseedge, bend left=25] (BX) to node[mylabel, below]{} (AX);
            \draw[myedge, bend left=135, black] (CX) to node[mylabel, left, pos=0.5]{+1} (AX);
            \draw[myreverseedge, bend left=50] (CX) to node[mylabel, left, pos=0.5]{} (AX);
            \draw[myedge] (AY) to node[mylabel, left, pos=0.4]{+1} (BY);
            \draw[myreverseedge, bend left=25] (AY) to node[mylabel, below]{} (BY);
            \draw[myedge, bend right=135, black] (CY) to node[mylabel, right, pos=0.5]{+1} (AY);
            \draw[myreverseedge, bend right=50] (CY) to node[mylabel, left, pos=0.5]{} (AY);
            \draw[myedge] (BX) to node[mylabel, right, pos=0.4]{+1} (CX);
            \draw[myreverseedge, bend right=25] (BX) to node[mylabel, below]{} (CX);
            \draw[myedge] (BY) to node[mylabel, above]{+1} (BX);
            \draw[myreverseedge, bend left=15] (BY) to node[mylabel, below]{} (BX);
            \draw[myedge] (CY) to node[mylabel, left, pos=0.4]{+2} (BY);
            \draw[myreverseedge, bend right=25] (CY) to node[mylabel, below]{} (BY);
            \draw[myedge] (CX) to node[mylabel, above]{+1} (CY);
            \draw[myreverseedge, bend right=15] (CX) to node[mylabel, below]{} (CY);

            \node[extlabel] at (AX.north east) [xshift=0pt, yshift=1pt] {2};
            \node[extlabel] at (AY.north west) [xshift=-3pt, yshift=1pt] {1};
            \node[extlabel] at (BX.north east) [xshift=0pt, yshift=1pt] {0};
            \node[extlabel] at (BY.north west) [xshift=-3pt, yshift=1pt] {2};
            \node[extlabel] at (CX.north east) [xshift=0pt, yshift=1pt] {1};
            \node[extlabel] at (CY.north west) [xshift=-3pt, yshift=1pt] {0};
        \end{tikzpicture}
    \end{subfigure}
\end{figure*}

\subsubsection{Harmonic}\label{app:zero_sum:harmonic}

A finite (cardinal) game is harmonic~\citep[Def. 1]{legacci2024no} when it admits a collection of action weights $\beta_{i,a_i} \in (0, \infty)$, $a_i \in \mathcal{A}_i$, $i \in [N]$, such that
\begin{align}
    \sum_i \sum_{a'_i \in \mathcal{A}_i} \beta_{i, a_i'} [ u_i(\va) - u_i(a_i', a_{-i}) ] = 0 \quad \forall \va \in \mathcal{A}.
\end{align}

In other words, there exists a weighting of the weighted response graph such that the net flow at every (joint action) node is zero. The condition above is linear in the action weights, so we can simply construct a matrix $A$ of the terms $u_i(\va) - u_i(a_i', a_{-i})$ and check if there exists weights $\beta = [\beta_{1A}, \beta_{1B}, \beta_{1C}, \beta_{2X}, \beta_{2Y}]$ that set $A\beta = \mathbf{0}$, i.e., check that $A$ is full column-rank (equiv., empty null-space). In the case of Figure~\ref{fig:adv_game}, the terms $u_i(\va) - u_i(a_i', a_{-i})$ are indicated in black next to the directed edges. The rows of $A$ are ordered $[AX, AY, BX, BY, CX, CY]$:
\begin{align}
    A &= \begin{bmatrix}
         0  &  2 &  1 &  0 & -1 \\
         0  & -1 &  1 &  1 &  0  \\
        -2  &  0 & -1 &  0 &  1  \\
         1  &  0 &  2 & -1 &  0  \\
        -1  &  1 &  0 &  0 & -1  \\
        -1  & -2 &  0 &  1 &  0
    \end{bmatrix}.
\end{align}
This matrix has full column-rank, therefore its null space contains only the zeros vector. So there does not exist an action weighting $\beta$ with $\beta_{i,a_i} \in (0, \infty)$ to render the net flow zero. Therefore, it is not harmonic.

\subsubsection{Preference-Zero-Sum}

While COGs do not admit a direct payoff structure, they do provide corresponding (\emph{unweighted}) response graphs. The COG in Figure~\ref{fig:adv_game} \emph{does} satisfy the definition of preference-zero-sum~\citep{biggar2023graph} which relies only on the unweighted response graph. A preference zero-sum game's corresponding response graph has only one sink component~\citep[Corollary 4.11]{biggar2023graph}, which we can calculate in time linear in the graph $\Theta(V + E)$, polynomial in the size of the game, using, e.g., Kosaraju-Sharir's algorithm~\citep[Algorithm p. 224 \& 229]{aho1983data}. The sink component conveys information about evolution, learning, and omits strictly dominated strategies. Nevertheless, we are unaware of any efficiency results for computing equilibria only assuming the preference-zero-sum property.

Hence, two-player, zero-sum does not appear to represent a complexity boundary in COGs similarly to classical NFGs.

\subsection{Two-Player, Known-Support}

We now consider instead the even simpler problem of computing an equilibrium of a $2$-player game given we know the support.

Probabilistic voting rules are sometimes avoided simply due to humans' aversion to a randomized election outcome. However, given our context-ordinal Nash equilibrium solution concept already allows for randomized play, it seems natural to consider them.

Here, we specifically consider maximal lotteries as the underlying voting rule of a $2$-player COG. Maximal lotteries (ML) are a probabilistic, Condorcet consistent voting rule, i.e., they select the action that beats every other action in a heads up comparison if one exists; otherwise, they specify a distribution over candidates that will be preferred to every other voting rule by a majority of voters in expectation. The maximal lottery problem can be formulated as a symmetric, \tpzs{} bi-matrix game ($\min_x \max_y x^\top A y$) where each entry in $A$, called the \emph{margin} matrix, equals the net frequency with which one candidate was ranked above another in the population of votes.
Here, we assume the NE has full support. Therefore, the maximum lottery best response must return a fully-mixed result.

In a COG, the margin matrix depends on the population of votes through the other player's strategy. Let the margin matrix be $A_i \in \mathbb{R}^{\vert \mathcal{A}_i \vert \times \vert \mathcal{A}_{i} \vert}$. $A_i$ here depends on $x_{-i}$, i.e., $A_i = A_i(x_{-i})$. The value of a symmetric, zero-sum game is zero and all actions in the support achieve this value at the NE, so we know that $A_i y = A_i x_i = u^{ML}_i(x_i) = \mathbf{0}$; we are looking for a symmetric NE of that game so $y = x = x_i$. It can be shown that
$A_i x_i = [x_{-i}^\top W_{i\ell} x_i]_{\ell} = \mathbf{0}$
where $W_{i\ell}$ is a constant matrix dependent only on the fixed preference data $\pref_i$ and $\ell \in [\vert \mathcal{A}_i \vert]$. Therefore, we are looking for an $x_i$ and $x_{-i}$ that satisfy
$x_{-i}^\top W_{i\ell} x_i = 0 \,\, \forall i,\ell$ subject to $x_{i}$ restricted to the given support for each player.
Empirically, we find that the resulting quadratic constraints are generally not convex (the relevant matrices are not positive semi-definite). This ultimately results in a system of quadratic equality constraints (QCQP), an NP-hard problem~\citep{pardalos1991quadratic}. Note a solution exists; we state QCQP complexity to express the general difficulty of computing its solution. %

\section{Metrics}\label{app:approx}

\subsection{Meta-Game Analysis Continued}

The meta-game analysis in Section~\ref{sec:approx:winrate} suggested measuring regret as the probability of electing a fixed hindsight action over the online algorithm for fixed regularization parameters $p$.

Alternatively, we can deploy our regularized best response to plot a curve showing the probability of selecting one of the algorithms ($x_{i,t}$ or $z_i^*$) as we vary the noise $p$ from $0$ to $1$ with $\mu_i$ as uniform and $q \ll 1$ held fixed. Area between the curve and a constant line at $0.5$ would provide a notion of how strong the selection bias is towards one algorithm or another. We present the probability of hindsight being selected by $\br_i^{(p,q=0.1,\mu_i=[\sfrac{1}{2},\sfrac{1}{2}])}$ in Figure~\ref{fig:regret_appx} along with this accompanying area metric in Figure~\ref{fig:auc_appx}.

\subsection{Game Space}\label{sec:approx:mov}

\begin{figure*}
    \centering
    \begin{subfigure}[t]{0.24\textwidth}
        \centering %
        \caption{Exploitability} %
        \label{fig:eps_appx} %
        \includegraphics[width=1.0\linewidth]{figures/eps_atari.png}
    \end{subfigure}
    \hfill
    \begin{subfigure}[t]{0.24\textwidth}
        \centering %
        \caption{Hindsight Winrate} %
        \label{fig:regret_appx} %
        \includegraphics[width=1.0\linewidth]{figures/prob_atari.png}
    \end{subfigure}
    \hfill
    \begin{subfigure}[t]{0.24\textwidth}
        \centering %
        \caption{AUC of~(\subref{fig:regret_appx})} %
        \label{fig:auc_appx} %
        \includegraphics[width=1.0\linewidth]{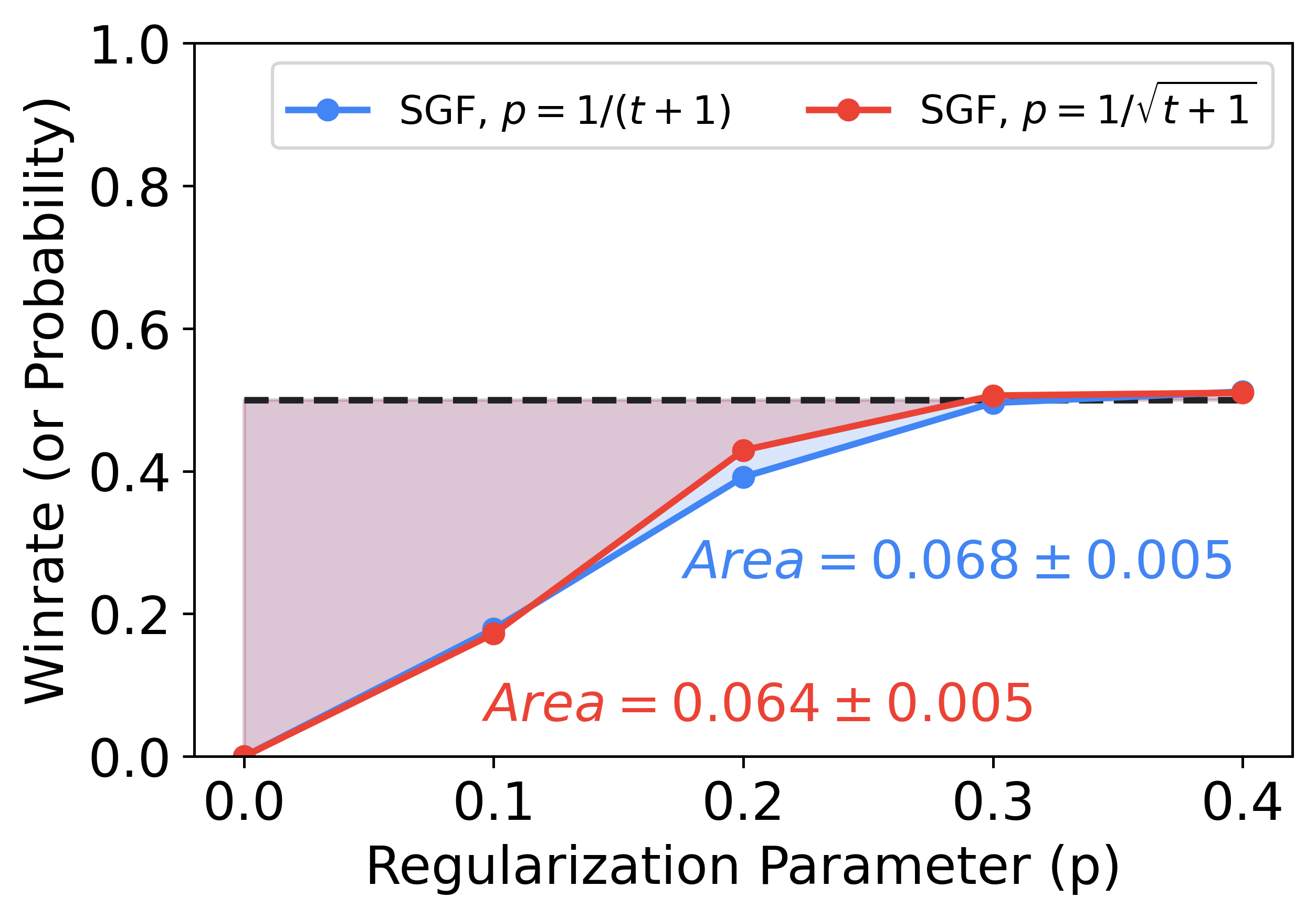}
    \end{subfigure}
    \hfill
    \begin{subfigure}[t]{0.24\textwidth}
        \centering %
        \caption{Hindsight MoV} %
        \label{fig:mov_appx} %
        \includegraphics[width=1.0\linewidth]{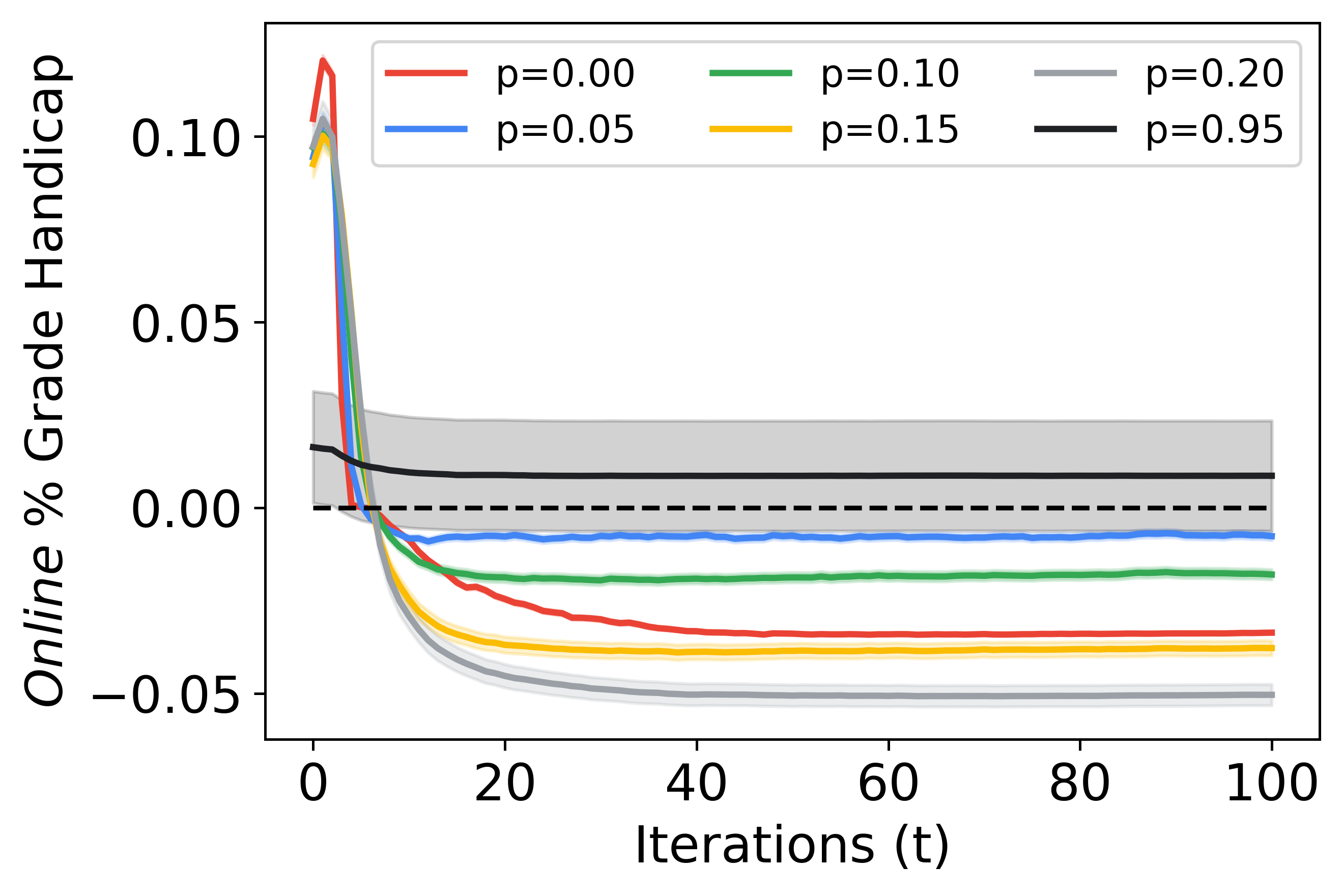}
    \end{subfigure}
    \caption{We evaluate our (SGF-based) FTRL approach ($p_s$ is the \emph{solver}'s parameter) on the Atari evaluation game according to the metrics in Section~\ref{sec:approx} and~\appxref{app:approx}. (\subref{fig:eps_appx}) EMD as defined in Section~\ref{sec:approx:strat}, eqn.~\peqref{eqn:emd}; (\subref{fig:regret_appx},\subref{fig:auc_appx}) Hindsight winrate and area measurement from Section~\ref{sec:approx:winrate}\textemdash (\subref{fig:auc_appx}) plots the \tikzcircle{2pt}'s seen in (\subref{fig:regret_appx}) at $t=100$ with $p$ on the $x$-axis; (\subref{fig:mov_appx}) margin of victory (MoV) as described in Section~\ref{sec:approx:mov}. Panels~(\subref{fig:regret_appx}) and~(\subref{fig:mov_appx}) set $p_s=\sfrac{1}{t+1}$. FTRL can also be considered a \emph{smoothed}-FP approach.}
    \label{fig:atari_regret_appx}
\end{figure*}

The classical approach \peqref{eqn:classical_eps_i} can be reinterpreted as taking the game as fixed, then finding the closest strategy $z$ to $x_i$ that achieves optimality (rationality) for player $i$, and finally reporting the difference between those two strategies in payoff space. Instead, we can take the strategy $x_i$ as fixed, find the ``closest'' game such that player $i$'s current strategy is rationalized, and then report the difference between those two games. This dual approach, also considered in other works~\citep{daskalakis2024charting,gemp2022d3c}, is actually analogous to the standard view taken in social choice. Specifically, \emph{margin of victory} counts the number of votes that must be altered for a given candidate to win the election.

In our lottery / infinite voter population model, it is easy to solve; we simply calculate what proportion of votes in the population need to be altered for candidates to tie. 
Moreover, this concept extends to the online setting via the meta-game approach described above by measuring how many votes must be altered in the vote population generated across all $T$ rounds. Figure~\ref{fig:mov_appx} displays this approximation metric.

\section{Algorithms}\label{sec:alg}

We discuss several practical learning algorithms inspired by work on normal-form games.
A best response operator regularized towards strategy $\mu_i$ with regularizer $R(z, \mu_i)$, $p_t > 0$,
\begin{align}
    \br_i^{(p_t,q,\mu_i)}([x_{-i,t}]) &= \argmax_{z}\{ -p_t R(z, \mu_i) + \bar{u}_{i,t}(z, [x_{-i,t}]) \} \label{eqn:br_reg}
\end{align}
forms the crux of many of these techniques including:
\begin{itemize}
    \item Follow the regularized leader~\citep{mcmahan2017survey,suggala2020follow} is a no-regret learning algorithm with time-average convergence to coarse-correlated equilibria~\citep{gordon2008no}; set $x_{i,0} = \argmin_z R(z, \mu_i)$; then
    \begin{align}
        x_{i,t+1}
        &= \argmax_{z}\{ -\frac{1}{t+1} R(z, \mu_i) + \sum_{t'=1}^t \frac{1}{t+1} u_{i,t'}(z, x_{-i,t}) \}
        \\ &= \argmax_{z}\{ -\frac{1}{t+1} R(z, \mu_i) + \bar{u}_{i,t}(z, [x_{-i,t}]) \}
    \end{align}
    which, c.f.~\peqref{eqn:br_reg}, suggests $\br^{(p_t,q,\mu_i)}(x_{-i,t})$ with $p_t = 1/(t+1)$;
    \item Fictitious play is a related algorithm in which each player best responds to the historical play of its co-players~\citep{brown1951iterative,robinson1951iterative}; smooth (regularized) variants enjoy regret guarantees~\citep{hofbauer2002global,fudenberg1995consistency,benaim2013consistency,benaim1999mixed},
    \item Homotopy~\citep{gemp2022sample} and adaptive regularization~\citep{sokotaunified,perolat2021poincare} methods all guide the algorithm through solving a curricula of games (using regularized best responses) that ends at the solution of the original game of interest.
\end{itemize}

And, in the case where a voting rule (e.g., scoring rule) induces a traditional normal-form game (see~\appxref{app:cog:nfg}), we can re-use any normal-form game solver as well. We examine both the follow the regularized leader (FTRL) and homotopy style of approaches in Section~\ref{sec:exp}.

\section{Visualization}\label{sec:viz}

The ability to visualize and analyze an equilibrium solution is important for diagnostics, particularly in game-theoretic evaluation~\citep{liure,marris2025deviation}. In these works, an action $a_i$'s rating is its expected payoff at equilibrium, and visualizations are given to show how each of another player's actions $a_j$ contribute to the rating of $a_i$. Here, we aim to provide analogous tools for COGs despite the lack of utility functions.

As mentioned earlier, some voting rules induce NFGs, for which we can reuse prior visualization and analysis techniques. For the more general case, we can consider breaking down the probability mass placed on an action $a_i$ at equilibrium into its contributions from each of player $j$'s actions. To do so, we leverage techniques in cooperative game theory, specifically Shapley values, although other power indices are possible. Shapley values satisfy an efficiency property that ensures the sum of the contributions from each action $a_j$ equals the probability mass placed on $a_i$ ($x_{i,a_i}$). The key primitive used in defining a Shapley value, and cooperative game theory general, is the characteristic function which acts on sets of actions. We define the characteristic function, $c_i$ applied to a subset of actions, $\hat{\mathcal{A}}_j \subseteq \mathcal{A}_j$, to be the probability mass placed on $a_i$ in player $i$'s best response when only actions $\hat{\mathcal{A}}_j$ are available to player $j$; we select the uniform distribution over winning actions as the unique best response value.
We define player $j$'s strategy over $\hat{\mathcal{A}}_j$ to be proportional to its original distribution over $\hat{\mathcal{A}}_j$ (i.e., normalized to sum to $1$), akin to omitting abstained votes.

For example in Chicken, the Shapley value breakdown of the final expected rank for the row player's actions at NE, $[0.5, 0.5]$, are:
\begin{align}
\begin{bmatrix}
 & $\texttt{swerve}$ & $\texttt{straight}$
\\ $\texttt{swerve}$ & $-0.25$ & $0.75$
\\ $\texttt{straight}$ & $0.75$ & $-0.25$
\end{bmatrix}.
\end{align}
When the row player swerves, if the column player swerves, the row player could have achieved a better outcome by going straight, so the column player's decision to swerve contributes negatively ($-0.25$) to the row player's rating of swerve. On the other hand, if the column player goes straight, the row player's decision to swerve achieves a much better outcome than if the row player had chosen to go straight which would have resulted in a collision.

\section{Experiments}

\subsection{Atari}\label{appx:atari}

In the Atari experiments, the two sources of stochasticity arise from a) the random usurper ballots and b) the Dirichlet sampling process described in Section~\ref{sec:reg}. Recall that this noise is used to regularize and render the best response operator continuous. We use this regularized best response operator $\br_i^{(p,q,\mu_i)}$ both as part of the FTRL-inspired update and to evaluate the learned strategy profile (see Section~\ref{sec:approx}). We fix a \texttt{solver\_seed} for the random noise generated for the FTRL solver as well as a separate \texttt{eval\_seed} to evaluate the solution returned by the solver. Every time the regularized best response operator is evaluated, we sample a best response according to the procedure described in Section~\ref{sec:approx} (and Algorithm~\ref{alg:reg_br}) \texttt{num\_samples} times; we then average the sampled best responses to give the regularized best response. Table~\ref{tab:atari_hyps} lists the hyperparameters used in the Atari experiments. Figure~\ref{fig:atari_regret} displays the mean and standard error over this set of random experiments.

\begin{table}[hb]
    \centering
    \begin{tabular}{c|c|c|c|c|c} \toprule
         & $q$ & $\mu_i$ & \# of \texttt{solver\_seed}s & \# of \texttt{eval\_seed}s & \texttt{num\_samples} \\ \midrule
        Figures~\ref{fig:eps} \&~\ref{fig:auc_appx} & $0.1$ & uniform & $10$ & $1000$ & $100$ \\
        Figures~\ref{fig:regret} \&~\ref{fig:mov_appx} & $0.1$ & uniform & $10$ & $100$ & $100$ \\ \bottomrule
    \end{tabular}
    \caption{Atari Hyperparameters.}
    \label{tab:atari_hyps}
\end{table}

\newpage
\subsection{Lost at Sea}\label{appx:lost_at_sea}

\subsubsection{Election instructions}\label{app:las_instructions}

The following reproduces the exact instructions on the election process provided to participants in the \textit{Lost at Sea} dataset from~\cite{qian2025maskmirror}.

\begin{mdframed}

Below is an overview of the election process.

\begin{enumerate}

\item \textbf{Indicating interest} - You will first be asked to indicate how much you want to become the group leader on a scale from 0 to 10.

\item \textbf{Ranking your teammates} - You will rank your three teammates, with your preferred leader at position 1, the second most preferred leader at position 2, and the third most preferred leader at position 3. You cannot vote for yourself.  

\end{enumerate}

We will use your answers to these two questions to select the leader:

\begin{itemize}

\item The two group members who express the most interest in becoming the leader will be selected as candidates for the election. If several group members choose the same number, the computer will randomly determine the order of these group members.

\item The highest-ranked candidate among the two will be elected as leader. If both candidates tie, the decision will be made randomly.
\end{itemize}

With this process, you are asked to rank your team members before knowing who the candidates are. Only the rankings of the two group members who are not candidates will be considered. This ensures that you cannot vote strategically to increase your own chances of being elected as the leader. Therefore, it is in the interest of all group members to provide their true, preferred ranking of the other group members.
\end{mdframed}

\subsubsection{Stochastic Outcomes}\label{app:chance}

To handle stochastic election outcomes, we consider the distribution defined by the following stochastic process. First sample a co-player action profile $a_{-i} \sim x_{-i}$. Then, independently sample a single (counterfactual) election outcome for every one of player $i$'s actions given $a_{-i}$. We then rank player $i$'s actions given that realized set of outcomes. We then repeat and aggregate these preferences using the chosen social choice voting rule.

\begin{figure}[h!]
    \centering
    \includegraphics[width=1.0\linewidth]{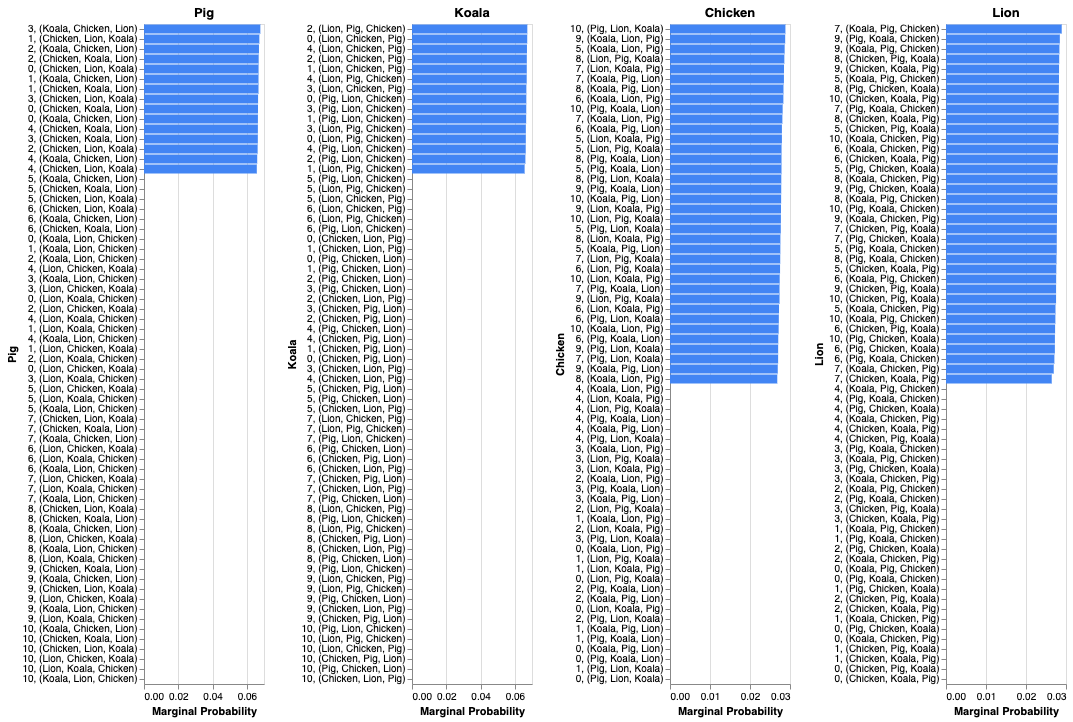}
    \caption{Election (Table~\ref{tab:not_pure_ne_group}) LLE Equilibrium.}
    \label{fig:las_game_b_eq}
\end{figure}

\subsubsection{Election Case Study: Human Play \emph{Not} in Equilibrium}

We examine the election data in Table~\ref{tab:not_pure_ne_group}, which contains both the actions the players took (\wtl{}s \& votes) as well as their \texttt{pref}s for election outcomes. We do not know whether the players' actions were sampled from a mixed strategy. For the sake of this analysis, we assume the players chose the actions listed in the table deterministically.
Under this assumption, calculations suggest the strategy profile listed in Table~\ref{tab:not_pure_ne_group} is not a maximal lottery (ML) NE.

In particular, \elephant{} is not playing a best-response. It is identified that if they reduced their \wtl{} to remove themselves from the runoff, then the two competing candidates would be \zebra{} and \turtle{} ranked 1st and 2nd by their preference rankings (\texttt{pref}). Across the players, \turtle{} wins once against \zebra{} (\dolphin{}'s vote) and loses once (\elephant{}'s vote). The winner is then either \turtle{} or \zebra{} with equal probability. This is compared to the status quo where \elephant{} enters a high \wtl{} such that \elephant{} and \zebra{} enter the runoff. \elephant{} beats \zebra{} twice (\dolphin{} and \turtle{} both rank it higher) so \elephant{} is deterministically elected, however, strangely, \elephant{} ranks itself 3rd in its own preference ranking (pref) despite its high \wtl{}.

Player \turtle{} is identified as not playing a best response by our sample best response estimate ($p=0$, $q=0.1$, $\texttt{num\_samples}=100$), but this is false. The sample best response we compute actually achieves the same value, but probabilistically. This is likely due to the sample estimate not having converged yet. The suggested best response is to submit a much higher \wtl{} (10) in order to force the runoff to be between \zebra{} and \turtle{}, ranked 0 and 2 by \turtle{}. They have equal wins so will be selected randomly giving \zebra{} an average score of 1 for the outcome regardless of the vote they pair with \wtl{}=10 since their vote cannot contain a comparison between \turtle{} (themselves) and \zebra{} by rules of the election.

Next, we demonstrate solving for an equilibrium of the voting game. In this demonstration, we use Borda as the voting rule. As mentioned earlier, Borda is a scoring rule which induces a normal-form game on the COG. We then re-use existing NFG solving techniques to approximate a limiting logit equilibrium (LLE)~\citep{mckelvey1995quantal}. We use the Jax~\citep{jax2018github} package \texttt{polarix} with \texttt{min\_temperature} $0.1$ to compute the LLE. Note that we use Borda here for efficiency sake, but it is possible to apply the same general technique to maximal lotteries with differentiable convex optimization libraries such as \textsc{cvxpylayers} in \textsc{JAX}~\citep{agrawal2019differentiable,jax2018github}.

In the LLE shown in Figure~\ref{fig:las_game_b_eq},
\turtle{} and \zebra{} spread most mass over high \wtl{} actions, so they are likely to be in the running. This means \dolphin{} and \elephant{}'s votes will have impact in the election.
Both \turtle{} and \zebra{} want themselves to win and are ranked favorably by \dolphin{} and \elephant{} although not in the same order. \dolphin{} prefers \turtle{}. \elephant{} prefers \zebra{}.
Expectedly, \dolphin{} ranks \turtle{} above \zebra{} in all its votes in the LLE support. \elephant{} ranks \zebra{} higher in all its votes. This should lead to \zebra{} and \turtle{} being elected with equal probability (confirmed empirically in simulation).

According to the players' outcome preferences (\texttt{pref}), the essential (bipartisan) set of a Maximal Lottery consists of \elephant{}, \turtle{}, and \zebra{}.
Why are \turtle{} and \zebra{} in the support of the LLE, but \elephant{} is not? \elephant{} actually prefers both \zebra{} and \turtle{} to itself, so it withdraws from the election to steer the result towards either of the two (preferably \zebra{}). The story is similar with \dolphin{}, who ranks themselves last.

\subsubsection{Election Case Study: Human Play in Equilibrium}

We similarly examine the election data in Table~\ref{tab:pure_ne_group}.
Calculations suggest each player's strategy listed in Table~\ref{tab:pure_ne_group} \emph{is} a max-entropy, maximal lottery best response to each other (i.e., a Nash equilibrium). We use the same settings as before  ($p=0$, $q=0.1$, $\texttt{num\_samples}=100$).

In particular, computing the max-entropy, maximal lottery for \bear{} reveals its best responding actions (essential set). We summarize \bear{}'s best response strategies as follows:
\begin{itemize}
    \item When \wtl{} $\le$ 6, \bear{} always ranks \camel{} over \goose{} because it is possible \bear{} is not selected in a runoff (\wtl{}=6), and therefore its vote (and \dog{}'s) matters. In this case, \camel{} wins, which is \bear{}'s ideal outcome;
    \item If \bear{}'s \wtl{} $>$ 6, then \bear{} is selected as the candidate against \camel{} and \bear{}'s vote doesn't matter (\dog{} and \goose{}'s do). \camel{} wins, which is \bear{}'s ideal outcome.
\end{itemize}

\begin{table}[t!]
    \centering
    \caption{Election A: (Pure) Strategies are in Equilibrium.}
    \label{tab:pure_ne_group}
    \begin{tabular}{l|c|c|c} \toprule
        voter & vote & \wtl{} & pref \\ \midrule
        \bearlabel{} & \ecamel{} $\succ$ \egoose{} $\succ$ \edog{} & 5 & \ecamel{} $\succ$ \egoose{} $\succ$ \ebear{} $\succ$ \edog{} \\
        \camellabel{} & \ebear{} $\succ$ \egoose{} $\succ$ \edog{} & 8 & \ecamel{} $\succ$ \ebear{} $\succ$ \edog{} $\succ$ \egoose{} \\
        \doglabel{} & \ecamel{} $\succ$ \ebear{} $\succ$ \egoose{} & 3 & \edog{} $\succ$ \ecamel{} $\succ$ \ebear{} $\succ$ \egoose{} \\
        \gooselabel{} & \ecamel{} $\succ$ \ebear{} $\succ$ \edog{} & 6 & \ecamel{} $\succ$ \ebear{} $\succ$ \egoose{} $\succ$ \edog{} \\ \bottomrule
    \end{tabular}
\end{table}

\paragraph{Borda LLE}

Computing an LLE of this election using Borda (Table~\ref{tab:pure_ne_group}), we find the following equilibrium description with equilibrium presented in Figure~\ref{fig:las_game_a_eq}.
\camel{} and \dog{} spread most of their mass over low \wtl{} actions, which removes them from the running. \bear{} and \goose{} spread their mass over high \wtl{} actions, so they are likely to be in the running.
\dog{} would prefer itself to win, but \dog{} is low-ranked by everyone. Given that others might not let \dog{} win the election, \dog{} submits a low \wtl{} to steer the election towards \bear{}, its third favorite, given \camel{} is not in the runoff.
\camel{} is high ranked by all but wants to avoid \goose{} being elected, which could happen if \camel{} and \goose{} are in a runoff. By submitting a low \wtl{}, \camel{} can influence the election away from \goose{} and towards \bear{} which is \camel{}'s 2nd ranked option.
\bear{} ranks itself lower than \goose{}, so would be happy with \goose{} winning. \goose{} essentially feels similarly, but with itself and \bear{} swapped in the ranking.
After sampling ten thousand election outcomes under the LLE, we observe that \bear{} is elected > 99\% of the time. It is interesting that the LLE arrives at \bear{} being elected because \camel{} is the strong Condorcet winner.

\begin{figure}[ht!]
    \centering
    \includegraphics[width=1.0\linewidth]{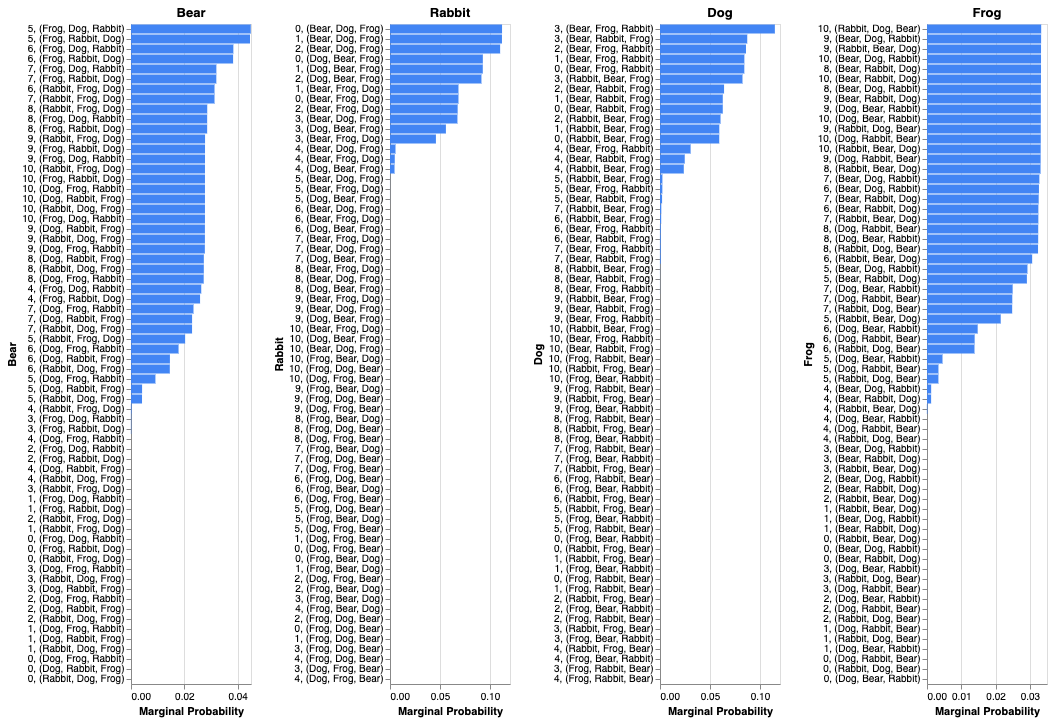}
    \caption{Election (Table~\ref{tab:pure_ne_group}) LLE Equilibrium.}
    \label{fig:las_game_a_eq}
\end{figure}

\end{document}